\documentclass[journal]{IEEEtranTCOM}
\normalsize
\usepackage[pdftex]{graphicx}
\usepackage[cmex10]{amsmath}
\usepackage{amsthm}
\usepackage{amssymb}
  \theoremstyle{plain}
  \newtheorem{theorem}{Theorem}
  \newtheorem{lemma}{Lemma}
\usepackage{url}
\usepackage{color}
\begin{document}
%  
% ****************************************************************************************
%  
\title{Constrained Bimatrix Games in  Wireless Communications}
\author{Koorosh~Firouzbakht,
        Guevara~Noubir,
        Masoud~Salehi% <-this % stops a space
        \thanks{K.~Firouzbakht and M.~Salehi (Emails: \{kfirouzb, salehi\}@ece.neu.edu) are with the Electrical and Computer Engineering Department, Northeastern University, Boston, MA 02115.}%
	\thanks{G.~Noubir (Email: noubir@ccs.neu.edu) is with the College of Computer and Information Science, Northeastern University, Boston, MA 02115. }%
	\thanks{Research partially supported by NSF Award CNS-0915985.}% <-this % stops a space
	\thanks{This paper was presented in parts at the 48th Asilomar Conference on Signal, Systems and Computers, 2014 and IEEE GLOBECOME 2014.}%
	\thanks{This work has been submitted to the IEEE for possible publication. Copyright may be transferred without notice, after which this version may no longer be accessible.}%
}
% 
% *** The paper running headers ***
% 
\markboth{Submitted to IEEE Transactions on Wireless Communications}%
{Shell \MakeLowercase{\textit{K.~Firouzbakht, G.~Noubir and M.~Salehi}}: Paper Title}
\maketitle
%
% The paper headers
% 
\markboth{Submitted to IEEE Transactions on Communications, 2015}%
{Submitted paper}
\maketitle
% 
% ****************************************************************************************
% 
\begin{abstract}
We develop a constrained bimatrix game framework that can be used to model many practical problems in many disciplines, including jamming in packetized wireless networks. 
In contrast to the widely used zero-sum framework, in bimatrix games it is no longer required that the sum of the players' utilities to be zero or constant, thus, can be used to model a much larger class of jamming problems.
Additionally, in contrast to the standard bimatrix games, in constrained bimatrix games the players' strategies must satisfy some linear constraint/inequality, consequently, not all strategies are feasible and the existence of the Nash equilibrium (NE) is not guaranteed anymore. 
We provide the necessary and sufficient conditions under which the existence of the Nash equilibrium is guaranteed, and show that the equilibrium pairs and the Nash equilibrium solution of the constrained game corresponds to the global maximum of a quadratic program. Finally, we use our game theoretic framework to find the optimal transmission and jamming strategies for a typical wireless link under power limited jamming. 
\end{abstract}
% 
% 
% ****************************************************************************************
% 
\section{Introduction}
\label{Sec_1_Introduction}
% 
% 
%  *** Introduction ***
% 

\IEEEPARstart{T}{he} 
convenience of wireless mobile communication has revolutionized the way we access information services and interact with the physical
world. Beyond enabling mobile devices to access information and data services ubiquitously,
wireless technology is  widely used in cyber-physical systems such as air-traffic control, power
plants synchronization, transportation systems, navigation systems and human body implantable devices. This
pervasiveness has elevated wireless communication systems to the level of critical infrastructure. Nevertheless, security issues of wireless communications remain a serious concern.

Physical layer in wireless networks is a broadcast medium that is subjected to adversaries. 
Among the many security threats that the wireless networks are subject to, jamming at the physical layer is one of the most prominent and challenging threats. 
Physical layer jamming not only can lead to service interruption/degradation or denial of service, but it is often a prelude
to other upper layer attacks such as spoofing, man in the middle and downgrade attacks~\cite{pelechrinis2011denial, vo2014cbm,  wu2007survey, vo2013counter, rengaraju2009analysis}.

Furthermore, many modern wireless networks such as sensor, ad-hoc and mesh networks often operate in a decentralized, self-configurable fashion. Network nodes are governed by a distributed protocol which allows the nodes to choose an action, i.e., make a \emph{decision}, from a set of available actions based on their evaluation of the network conditions (possibly relying on the information provided by the other nodes). These decisions not only have impact on the performance of individual nodes but may have impact on the overall performance of the entire network. 
Nodes may seek the \emph{greater good} of the network that is, they seek actions that optimize the overall performance of the entire network, or they can act \emph{selfishly} and compete with other nodes to optimize their individual performance. 
Additionally, nodes may act \emph{maliciously}, i.e,
seek actions that result in performance degradation of the individual nodes or the entire network.

All these examples have many of the characteristics that would lead to a natural game theoretic formulation as these problems cannot be completely modeled by the traditional optimization tools. Moreover, as software defined radios (SDRs) and Cognitive Radios (CRs) become more capable of implementing more sophisticated and complicated adaptation algorithms, the assumptions of game theoretic models become an even better match for future wireless networks.

Game theory has been used to solve problems in numerous aspects of wired and wireless communication systems --- from security at the physical and MAC layer (e.g., jamming and eavesdropping) to routing and intrusion detection systems (e.g., collaborative IDS's) at upper layers of the protocol stack. 
\cite{manshaei2013game} and the references therein provide a structured and comprehensive overview of the game-theoretic approach to security and privacy in computer and communications networks. In \cite{Charilas_2010}, the authors use a layered approach to survey applications of game theory in wireless networking and game models that are most suitable for each problem.

Reference \cite{Liang_2013} presents a classification of applications of game theory in network security based on the game model that is used to approach the problem. The survey covers both cooperative and non-cooperative games (see Table I in \cite{Liang_2013}).
In what follows, we briefly review some of the applications of game theory in wireless communications. We limit our focus to the physical layer applications, 
for applications of game theory in other network layers we refer the reader to~\cite{manshaei2013game} and~\cite{Charilas_2010}.

% 
%  ******************************************
% 
\subsubsection{Dynamic Spectrum Access (DSA)}
addresses the issue of how to allocate the limited available spectrum among multiple wireless devices. This problem has two important aspects, spectrum usage efficiency and fairness to wireless users and can be modeled as a cooperative or non-cooperative game. Cooperative models such as \emph{bargaining games} and \emph{coalition games} are often used where cooperation between spectrum users can result in an equilibrium that is more efficient and fair~\cite{ji2007cognitive, saad2009coalitional}.

On the other hand, when collaboration between network nodes is not possible or permitted for instance, when secondary users compete for channel or spectrum access in cognitive radio networks, non-cooperative frameworks (e.g. \emph{auction-based games}) are used to model this problem~\cite{wang2010game}. 
In cognitive radio networks, the wireless spectrum is shared between primary users (users that are licensed to operate in that spectrum band) and secondary users (unlicensed users) 
where the secondary users access the spectrum in an \emph{opportunistic manner}. In this scenario, the secondary users compete with each other to access the bandwidth offered by the primary users and therefore; spectrum sharing in better modeled under the non-cooperative framework \cite{huang2006auction, wu2009scalable}.

% 
%  ******************************************
% 
\subsubsection{Power control in CDMA networks}
is another example that has been studied by game theoretic approaches~\cite{meshkati2006game, chiang2008power, menon2009game}. In CDMA power control problem, a player's utility function is usually defined such that it increases with the signal to interference plus noise ratio (SINR) but decreases with transmission power~\cite{goodman2000power}. This assumption is well justified since an increase in the SINR results in lower error probability and hence would increase the quality of service. Assuming all network nodes operate at fixed transmission power, an increase in a node's transmission power, results in higher SINR for that node. However, increasing transmission power results in lower SINR for other users. Because most wireless nodes are battery-operated, energy management is an important consideration and if a node's power is too high, not only it reduces other nodes' SINR but it also wastes valuable battery life.

% 
%  ******************************************
% 
\subsubsection{Power control in OFDMA networks}
Game theory has also been used to study power control in Orthogonal Frequency Devision Multiple Access (OFDMA) networks. In OFDMA networks, the objective is to minimize the overall transmission power under rate and power constraints by allocating users' rates and powers to the available sub-channels.

This problem has been studied in the literature in non-cooperative \cite{han2007non, jing2009distributed} and cooperative frameworks \cite{yang2010optimal}. In non-cooperative framework network nodes use local and selfish power control strategies to maximize their individual performances while in the cooperative framework, network nodes use distributed and (possibly) selfless power control strategies to optimize the overall performance and fairness of the system. It has been shown that game theoretic power control strategies can achieve significant individual and/or overall performance improvements over traditional power control algorithms~\cite{meshkati2006game}.

% 
%  ******************************************
% 
\subsubsection{The jamming problem} among the many security threats that the wireless networks are subject to, jamming at the physical
layer is one of the most prominent and challenging threats. 
Jamming at the physical layer is often modeled as a zero-sum game, a special class of non-cooperative games~\cite{pelechrinis2010efficacy, zhou2011optimizing, wang2011anti, alpcan2011security, Koorosh_2013}.
In a zero-sum game, for all strategy profiles the sum of players' payoffs is zero and as a result, if a player gains a payoff, that payoff must have been lost by other player(s). Jamming in wireless communication is one such case as the players, the communicating nodes and the jammer, have completely conflicting goals.

In zero-sum game framework, it is usually assumed that the players have \emph{perfect knowledge} of the game and the actions that are available to the other players, and they use this knowledge to compute their respective optimal strategies. 
In such a case, the zero-sum framework fully captures the conflicting goals of the players. Moreover, the equilibrium solution of the zero-sum game 
guarantees a minimum payoff regardless of the other player's strategy~\cite{Owen}.

% 
%  ******************************************
% 
\begin{figure}
	\centering
	\includegraphics[width = .6\linewidth]{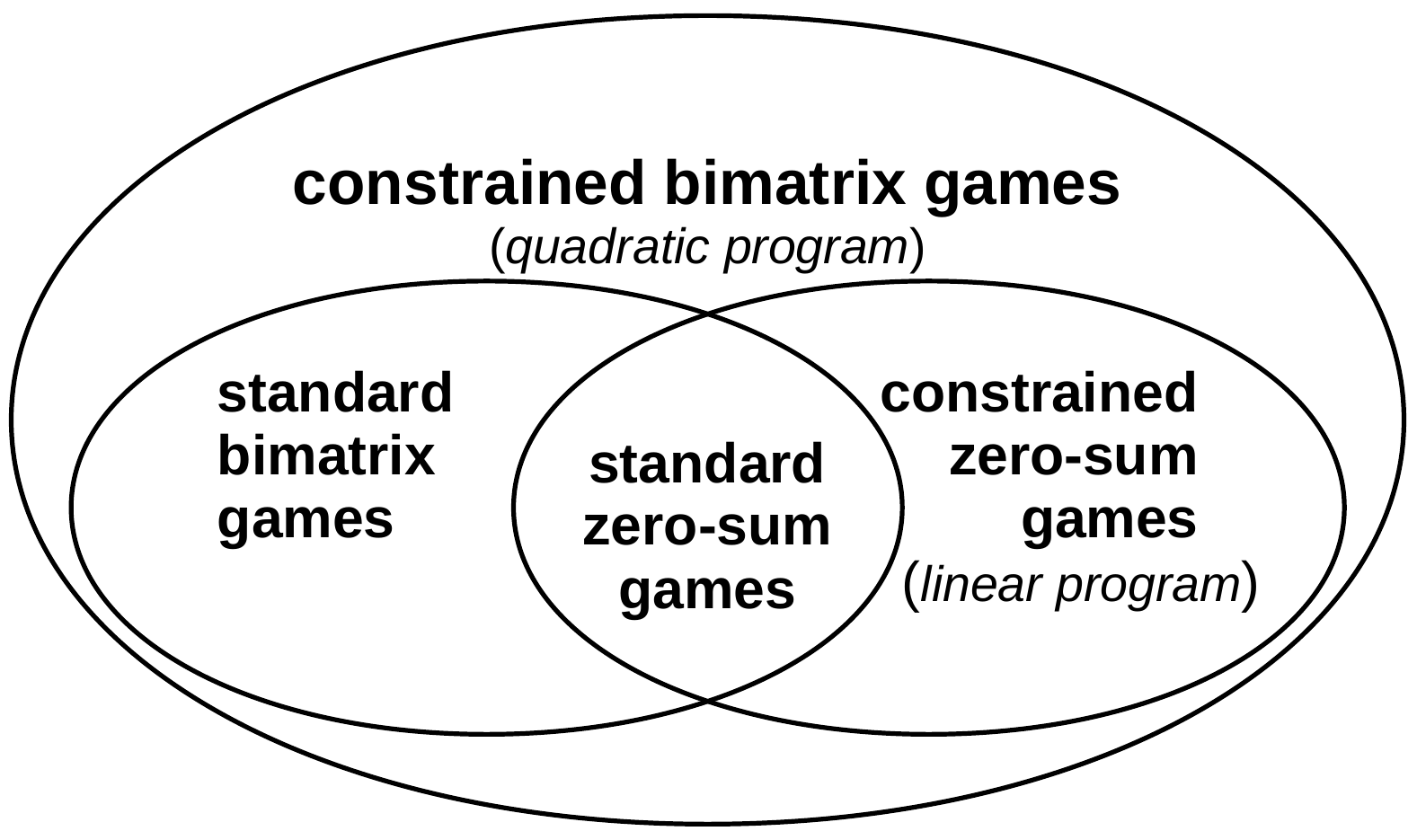}
	\caption{Classification of standard and constrained games.}
	\label{Fig_Bimatrix_and_Zerosum_Relation}
\end{figure}
% 
%  ******************************************
% 

However, in some jamming scenarios, having perfect knowledge of the system parameters (or available actions) may not be a feasible option or too costly for a player. In addition, players may have objectives that are not exactly the opposite of each other, for example, the transmitter may wish to minimize the average error probability while the jammer wishes to minimize the average throughput of the system (as opposed to maximizing the average error probability). 

In such scenarios, a more appropriate framework to model the communication system under jamming would be a  \emph{bimatrix} game instead of a zero-sum game%
\footnote{It can be shown that zero-sum games ate special cases of the more general bimatrix games.}. 
In bimatrix games it is no longer required that the sum of the players' payoffs to be zero (or a constant value)~\cite{Owen}. As a result, players can have different objectives and the respective payoffs can be defined based on the players' goals and their knowledge of the game (which in general may be imperfect). 
Such a formulation, encompasses a variety of situations from full competition to full cooperation.

Additionally, in standard zero-sum and bimatrix games there are no additional restrictions on players' mixed-strategies, i.e., players may choose any probability distribution over their respective action sets (pure-strategies). 
However, there exist scenarios where, due to practical reasons, not all mixed-strategies are permitted and/or feasible.

Such scenarios demand for a more general framework to study them.
In this paper, we study a \emph{constrained bimatrix} game to overcome these limitations. In constrained games, the players' mixed-strategies not only have to be a probability distributions but they must satisfy some additional constraints too
(Figure~\ref{Fig_Bimatrix_and_Zerosum_Relation} shows the classification of standard and constrained games). We study the necessary and sufficient conditions under which the existence of the Nash equilibrium is guaranteed as well as a systematic approach to find the NE.

The rest of the paper is organized as follows, in Section~\ref{Sec_2_Constrained_Bimatrix_Games} we will introduce the \emph{constrained bimatrix} framework
and provide the necessary and sufficient conditions under which the existence of a constrained NE solution is guaranteed.
In Section~\ref{Sec_3_Quadratic_Programming} we show that the solution of the this constrained game corresponds to the global maximizers of a \emph{quadratic program}. In Section~\ref{Sec_4_Case_Study} we will use the framework that we developed to study a typical jamming problem. Finally, we conclude the paper in Section~\ref{Sec_5_Conclusion}.

% 
% ****************************************************************************************
% 
\section{Constrained Bimatrix Games}
\label{Sec_2_Constrained_Bimatrix_Games}
We start by introducing the concept of the Nash equilibrium (NE) for standard bimatrix games. 
Then, we generalize the standard bimatrix framework by adding linear constraints on the players strategies and formulate the \emph{constraint bimatrix} framework. 
We refer the reader to \cite{Felegyhazi_2006}  and the references therein for an introduction to some of the most fundamental concepts of non-cooperative game theory. This tutorial is specifically written for wireless network engineers and uses intuitive examples that are focused on wireless networks. \cite{Mackenzie_2006} provides a more comprehensive review of non-cooperative game theory and its applications in wireless communications and networking.
Throughout the rest of the paper, we refer to player one (row player) as the \emph{transmitter} and refer to player two (column player) as the \emph{jammer}. Nevertheless, applications of our framework is not limited to jamming in wireless communications.

Consider a bimatrix game where transmitter's action set (for instance, transmission rates) is given by
\begin{equation}
\label{Eq_Player_1_Action_Set}
 \mathcal{R} = \big\{ r_1,r_2, \cdots, r_m \big\} \qquad r_i \in \mathbb{R}_+,\ 1\leq i \leq m
\end{equation}
Without loss of generality assume $\mathcal{R}$ is a sorted set, i.e., $0 \leq r_1 < \cdots < r_{i-1} < r_i < \cdots < r_m$. Transmitter's vector of possible actions (simply, \emph{action vector}) is the column vector $\mathbf{r}$ defined as
\begin{equation}
\label{Eq_Player_1_Action_Vector}
 \mathbf{r}^T = \left[r_1 \ r_2 \ \cdots \ r_i \ \cdots r_m \right]_{1\times m}
\end{equation}
where $^T$ indicates matrix transposition. 
Similarly, we define the jammer's action set and action vector as
\begin{equation}
\label{Eq_Player_2_Action_Set}
 \mathcal{J} = \big\{ j_1,j_2, \cdots, j_n \big\} \qquad j_k \in \mathbb{R}_+,\ 1\leq k \leq n
\end{equation}
and
\begin{equation}
\label{Eq_Player_2_Action_Vector}
 \mathbf{j}^T = \left[j_1 \ j_2 \ \cdots \ j_k \ \cdots j_n \right]_{1\times n}
\end{equation}
where WLOG we assume, $0 \leq j_1 < \cdots < j_{k-1} < j_k < \cdots < j_n$. 
A \emph{standard bimatrix game} (also known as two-player general sum game) is defined by a pair of $m\times n$ matrices $A$ and $B$ such that, if player one plays row $i$ and player two plays column $k$, the elements at row $i$ and column $k$ of the matrices $A$ and $B$  (i.e., $a_{ik}$ and $b_{ik}$) would be the payoffs received by players one and two, respectively.

If we let the players randomize their actions (i.e., allow them to use mixed-strategies\footnote{A player is playing a mixed-strategy if he randomizes his actions over his action set according to a probability distribution.}), the expected payoffs of the game for the mixed strategy profile $(\mathbf{x}, \mathbf{y})$ are
\begin{equation}
 \begin{array}{ll}
   A \big( \mathbf{x}, \mathbf{y} \big) \triangleq \mathbf{x}^T A \ \mathbf{y} &  \qquad \text{for player I} \\
   B \big( \mathbf{x}, \mathbf{y} \big) \triangleq \mathbf{x}^T B \ \mathbf{y} & \qquad  \text{for player II} \\
 \end{array}
\end{equation}
where $\mathbf{x}\in \mathbf{X}^m \triangleq \big\{\mathbf{x} \in \mathbb{R}^m_+ \big| \sum_{i=1}^m x_i = 1 \big\}$ and $\mathbf{y}\in \mathbf{Y}^n \triangleq \big\{\mathbf{y} \in \mathbb{R}^n_+ \big| \sum_{k=1}^n y_k = 1 \big\}$ are mixed-strategy vectors of players one and two, respectively.
Player one's goal is to find an optimal strategy, $\mathbf{x}$, that maximize his expected payoff (given by payoff matrix $A$) against player two's strategy, $\mathbf{y}$, i.e., player one wants to solve the following problem
\begin{equation}
\label{Eq_Player_1_Goal}
 \underset{\mathbf{x} \in \mathbf{X}^m}{\text{maximize}}\ A \big( \mathbf{x}, \mathbf{y} \big) 
 \quad 
 \text{for all}\quad \mathbf{y} \in \mathbf{Y}^n
\end{equation}
while player two's goal is to maximize his own payoff (given by payoff matrix $B$) by solving the following problem
\begin{equation}
\label{Eq_Player_2_Goal}
 \underset{\mathbf{y} \in \mathbf{Y}^n}{\text{maximize}}\ B \big( \mathbf{x}, \mathbf{y} \big) 
 \quad 
 \text{for all}\quad \mathbf{x} \in \mathbf{X}^m
\end{equation}
The strategy profile $\big( \mathbf{x}^*, \mathbf{y}^* \big)$ is said to be an equilibrium pair (or equivalently the Nash equilibrium, NE) if $(\mathbf{x}^*, \mathbf{y}^*)$ satisfies \eqref{Eq_Player_1_Goal} and \eqref{Eq_Player_2_Goal} simultaneously. 
That is, $\mathbf{x}^*$ maximizes~\eqref{Eq_Player_1_Goal} for $\mathbf{y}^*$ and $\mathbf{y}^*$ maximizes~\eqref{Eq_Player_2_Goal} for $\mathbf{x}^*$,
and therefore, no player benefits by unilaterally changing his strategy.

\begin{theorem}
\label{Theorem_1}
 Every finite bimatrix game in its standard form has at least one equilibrium pair (Nash equilibrium) in mixed-strategies.
\end{theorem}
\begin{proof}
 See Theorem 1 in \cite{Nash51}.
\end{proof}

%  *** Table ***>
%  *****************************************************
% 
\begin{table*}[t]
\centering \caption{Necessary and Sufficient Conditions for a Pair $(\mathbf{x}^*, \mathbf{y}^*)$ to Be the NE.}
\label{Table_KKT_Conditions_for_Bimatrix_Game}
\renewcommand{\arraystretch}{1.5}
 \begin{tabular}[t]{lr||lr}
  \multicolumn{2}{c}{ Player I } &\multicolumn{2}{c}{Player II} \\
  \hline
  \hline 
  $\mathbf{1}^T \mathbf{x}^* - 1 = 0$ 			&(I.1)&$ \mathbf{1}^T \mathbf{y}^* - 1 = 0 $&(II.1)\\
  $\mathbf{r}^T \mathbf{x}^* - r_{\mathrm{ave}} \leq 0$	&(I.2)&$ \mathbf{j}^T \mathbf{y}^* - j_{\mathrm{ave}} \leq 0$ &(II.2)\\
  $-\mathbf{x}^* \leq 0$ 					&(I.3)&$ -\mathbf{y}^* \leq 0$ &(II.3)\\
  \hline
  $A\mathbf{y}^* - u \mathbf{r}_{m\times1} - \alpha \mathbf{1}_{m\times1} \leq \mathbf{0}$ &(I.4)& 
  ${\mathbf{x}^*}^T B - v \mathbf{J}_{n\times1} - \beta \mathbf{1}_{n\times1} \leq \mathbf{0}$ &(II.4)\\
  \hline
  ${\mathbf{x}^*}^T A\ \mathbf{y}^* - u r_{\mathrm{ave}} - \alpha =0$ &(I.5)&
  ${\mathbf{x}^*}^T B\ \mathbf{y}^* - v j_{\mathrm{ave}} - \beta =0$ &(II.5)\\
  $u\left( \mathbf{r}^T \mathbf{x}^* - r_{\mathrm{ave}} \right) = 0$ &(I.6)& 
  $v\left( \mathbf{j}^T \mathbf{y}^* - j_{\mathrm{ave}} \right) = 0$ &(II.6)\\
  \hline
  $u \geq 0, \ \alpha \in \mathbb{R}$ 			&(I.7)& $v \geq 0, \ \beta \in \mathbb{R}$ &(II.7)\\
 \end{tabular}
\end{table*}
% 
%  *****************************************************
% 
 
Consider a bimatrix game for which, due to practical reasons, not all mixed-strategies are permitted and/or are feasible. 
For instance, assume maximizing the average throughput of a wireless link.
Maximizing the average throughput requires using higher transmission rates; but to maintain an acceptable error rate at the receiver, higher rates must be transmitted at higher transmission power. Because of battery limitation (internal limitation) or the FCC regulations (external limitations), the transmitter must keep its average transmission power below a certain value. Consequently, the wireless user cannot use certain actions that are more preferable to him (such as transmitting at the highest rate all the time). He may only choose actions that result in an average transmission power less than or equal to a predetermined value. 

Assume the mixed-strategy pair $\mathbf{x}$ and $\mathbf{y}$ must be chosen from some hyperpolyhedron defined by linear inequalities,
\begin{equation}
\label{Eq_Player_1_Constrained_Strategy_Set}
 \mathbf{x} \in \widehat{\mathbf{X}} \triangleq \big\{ \mathbf{x}\in \mathbf{X}^m \ \big| \ \mathbf{r}^T \mathbf{x} \leq r_{\mathrm{ave}} \big\}
\end{equation}
and
\begin{equation}
\label{Eq_Player_2_Constrained_Strategy_Set}
 \mathbf{y} \in \widehat{\mathbf{Y}} \triangleq \big\{ \mathbf{y} \in \mathbf{Y}^n \ \big| \ \mathbf{j}^T \mathbf{y} \leq j_{\mathrm{ave}} \big\}
\end{equation}
we denote this constraint game by $\mathcal{G} = \big( A,B, $ $\mathbf{r}, \mathbf{j}, r_{\mathrm{ave}}, j_{\mathrm{ave}} \big)$ where $A,B \in \mathbb{R}^{m\times n}$ are the payoff matrices for player one and two, respectively, and weight vectors $\mathbf{r}$ and $\mathbf{j}$ are given by \eqref{Eq_Player_1_Action_Vector} and \eqref{Eq_Player_2_Action_Vector}, respectively%
\footnote{This is different from the bimatrix game in its standard form where the payoff matrices $A$ and $B$ completely define the game.}%
.
It is easily verified that for $r_{\mathrm{ave}} \geq \max r_i$ and $j_{\mathrm{ave}} \geq \max j_k$ the constrained game simplifies to the bimatrix game in its standard form; hence, the unconstrained game can be viewed as a special case of the constrained games (see Figure~\ref{Fig_Bimatrix_and_Zerosum_Relation}).
Therefore, in the following we assume that at least one of the following inequalities holds
\begin{equation}
\label{Eq_Average_Constraint}
 r_{\mathrm{ave}} < \max r_i \quad \text{or} \quad j_{\mathrm{ave}} < \max j_k
\end{equation}

By introducing \eqref{Eq_Player_1_Constrained_Strategy_Set} and \eqref{Eq_Player_2_Constrained_Strategy_Set}, and assuming that \eqref{Eq_Average_Constraint} holds, we are eliminating some mixed-strategies that could have been otherwise selected. 
Therefore, the existence of the NE solution for this constrained bimatrix game is not trivial and must be established (see  Appendix).

Assuming that in the constrained bimatrix game $\mathcal{G}=$ $\big( A,B, $ $\mathbf{r}, \mathbf{j},$ $ r_{\mathrm{ave}}, j_{\mathrm{ave}} \big)$, the jammer is playing his optimal strategy $\mathbf{y}^*$. Transmitter's optimal strategy, $\mathbf{x}^*$, against $\mathbf{y}^*$ is the maximizer of the following problem
\begin{equation}
\label{Eq_Player_1_Constrained_Problem}
 \underset{\mathbf{x}}{\text{maximize}}\ \mathbf{x}^T A\ \mathbf{y}^* 
 \quad 
 \text{s.t.} \ 
 \left\{
 \begin{array}{l}
  \mathbf{1}^T \mathbf{x} - 1 = 0 			\\%&\quad \#1 & \mu \\
  \mathbf{r}^T \mathbf{x} - r_{\mathrm{ave}} \leq 0 	\\%&\quad \#1 & u\\
  - \mathbf{x} \leq 0 					\\%&\quad \#m & \lambda_i\\
 \end{array}
 \right.
\end{equation}
Similarly, jammer's optimal strategy, $\mathbf{y}^*$, against $\mathbf{x}^*$ is the maximizer of the following problem
\begin{equation}
\label{Eq_Player_2_Constrained_Problem}
 \underset{\mathbf{y}}{\text{maximize}}\ {\mathbf{x}^*}^T B\ \mathbf{y} 
 \quad 
 \text{s.t.} \ 
 \left\{
 \begin{array}{l}
  \mathbf{1}^T \mathbf{y} - 1 = 0 			\\%&\quad \#1 & \mu \\
  \mathbf{j}^T \mathbf{y} - j_{\mathrm{ave}} \leq 0 	\\%&\quad \#1 & u\\
  - \mathbf{y} \leq 0 					\\%&\quad \#n & \lambda_i\\
 \end{array}
 \right.
\end{equation}

Individually, \eqref{Eq_Player_1_Constrained_Problem} and \eqref{Eq_Player_2_Constrained_Problem} are linear programs, but $\mathbf{x}^*$ and $\mathbf{y}^*$ are not known in advance or in general they may not even exist.
Theorem~\ref{Theorem_2} gives the necessary and sufficient conditions that any Nash equilibrium solution of $\mathcal{G}$ must satisfy.
That is, every NE solution of $\mathcal{G}$ satisfies the conditions in Table~\ref{Table_KKT_Conditions_for_Bimatrix_Game} and every strategy pair $(\mathbf{x}, \mathbf{y})$ that satisfies the conditions in Table~\ref{Table_KKT_Conditions_for_Bimatrix_Game} must be a NE.
Additionally, in the Appendix%~\ref{Appendix_I}
, we prove the conditions under which existence of the NE for this constrained bimatrix game is guaranteed.

\begin{theorem}
\label{Theorem_2}
  Let $\mathcal{G} = \big( A,B, $ $\mathbf{r}, \mathbf{j}, r_{\mathrm{ave}}, j_{\mathrm{ave}} \big)$ be a  constrained bimatrix game defined by matrices $A,B \in \mathbb{R}^{m\times n}$. A strategy pair $(\mathbf{x}^*, \mathbf{y}^*)$ is an equilibrium pair (NE), if and only if there exists scalers $u,v \geq 0$ and $\alpha, \beta \in \mathbb{R}$ such that the conditions in Table \ref{Table_KKT_Conditions_for_Bimatrix_Game} are satisfied.
\end{theorem}
\begin{proof}
Consider the KKT conditions for the linear program \eqref{Eq_Player_1_Constrained_Problem}. 
The optimal solution $\mathbf{x}^*$ must satisfy the \emph{primal feasibility} conditions given by %of the linear program, i.e.,
\begin{equation}
%  \text{Primal Feasibility:} 
%  \
 \left\{
  \begin{array}{l}
  \mathbf{1}^T \mathbf{x}^* - 1 = 0 			\\
  \mathbf{r}^T \mathbf{x}^* - r_{\mathrm{ave}} \leq 0 	\\
  - \mathbf{x}^* \leq 0 					\\
 \end{array}
 \right.
\end{equation}
which are identical to conditions (I.1) -- (I.3) in Table~\ref{Table_KKT_Conditions_for_Bimatrix_Game}. 
From the \emph{dual feasibility} conditions we must have
\begin{equation}
\label{Eq_Dual_Feasibility}
%  \text{Dual Feasibility:}
%  \quad
 \begin{aligned}
  \nabla \Big( \mathbf{x^*}^T A\ \mathbf{y}^* \Big) &- 
   \sum_{i=1}^{m}\lambda_i \nabla({-x_i^*})  \\
    & - u\nabla \big( \mathbf{r}^T \mathbf{x}^* - r_{\mathrm{ave}} \big) \\ 
    & - \mu \nabla \big( \mathbf{1}^T \mathbf{x}^* - 1 \big) = 0\\
 \end{aligned}
%  \text{such that}
%  \left|
%  \begin{array}{l}
%   \lambda_i \geq 0 \\%,  1\leq i \leq m\\
%   u \geq 0 \\
%   \mu \in \mathbb{R} \\
%  \end{array}
%   \right.
\end{equation}
such that
\begin{equation}
    \lambda_i \geq 0, 	\quad
      u \geq 0,		\quad
      \mu \in \mathbb{R}
\end{equation}
where $\lambda_i$, $u$ and $\mu$ are the KKT multipliers corresponding to constraints in \eqref{Eq_Player_1_Constrained_Problem}.
If we simplify \eqref{Eq_Dual_Feasibility} and use vector representations for KKT multipliers we get
\begin{equation}
\label{Eq_Dual_Feasibility_Simplified}
 A\mathbf{y}^* + \boldsymbol{\lambda}_{m\times1}- u \mathbf{r} - \mu \mathbf{1}_{m\times1}=0 
\end{equation}
or, equivalently,
\begin{equation}
% \boxed{%
 A\mathbf{y}^* - u \mathbf{r} - \mu \mathbf{1}_{m\times1} \leq \mathbf{0} \qquad u\geq 0 \ \text{and}\ \mu \in \mathbb{R}
%  }
\end{equation}
which gives us  condition (I.4) in Table \ref{Table_KKT_Conditions_for_Bimatrix_Game} (where we have made a change of variable, $\mu \rightarrow \alpha$, and have used the fact that $\boldsymbol{\lambda}_{m\times1} \geq \mathbf{0}$). Finally, from the \emph{complementary slackness} conditions we must  have
\begin{equation}
\label{Eq_Complementary_Slackeness_for_Player_1}
%  \text{Complementary Slackness}
%  \quad
 \left\{
 \begin{array}{l}
  \boldsymbol{\lambda}^T \mathbf{x}^* = 0 \\
   u\left( \mathbf{r}^T \mathbf{x}^* - r_{\mathrm{ave}} \right) = 0 \\
 \end{array}
  \right.
\end{equation}
the second condition in \eqref{Eq_Complementary_Slackeness_for_Player_1} is identical to (I-6). By multiplying \eqref{Eq_Dual_Feasibility_Simplified} by ${\mathbf{x}^*}^T$ and using $\boldsymbol{\lambda}^T \mathbf{x}^* = 0$ we have
\begin{equation}
 {\mathbf{x}^*}^T A\ \mathbf{y}^* + {\mathbf{x}^*}^T \boldsymbol{\lambda}- u {\mathbf{x}^*}^T \mathbf{r} - \mu {\mathbf{x}^*}^T \mathbf{1}=0
\end{equation}
which can be further simplified to
\begin{equation}
%  \boxed{
  {\mathbf{x}^*}^T A\ \mathbf{y}^* - u r_{\mathrm{ave}} - \mu =0 \qquad u\geq 0 \ \text{and}\ \mu \in \mathbb{R}
%  }
\end{equation}
which results in condition (I.5) in Table \ref{Table_KKT_Conditions_for_Bimatrix_Game}. In the exact same way, we can derive KKT's necessary conditions of optimality for the jammer to get the conditions (II.1) -- (II.7) in Table~\ref{Table_KKT_Conditions_for_Bimatrix_Game}. To prove that these conditions are also sufficient, we can use the fact that the objective functions in \eqref{Eq_Player_1_Constrained_Problem} and \eqref{Eq_Player_2_Constrained_Problem} are linear (affine), and as a result, the KKT conditions are necessary and sufficient for optimality; this concludes the proof.
\end{proof}
Table \ref{Table_KKT_Conditions_for_Bimatrix_Game} summarizes the necessary and sufficient conditions of optimality for the constrained bimatrix game. Furthermore, the following lemma, gives  the expected payoff of the players at the Nash equilibrium.
\begin{lemma}
\label{Lemma_Bimatrix_Game_Expected_Payoffs}
 Consider the constrained bimatrix game $\mathcal{G} = \big( A,B, $ $\mathbf{r}, \mathbf{j}, r_{\mathrm{ave}}, j_{\mathrm{ave}} \big)$. The expected payoffs of the game for the equilibrium pair $(\mathbf{x}^*, \mathbf{y}^*)$ are
 \begin{align}
  &	A\big( \mathbf{x}^*, \mathbf{y}^* \big) = u r_{\mathrm{ave}} + \alpha \\
  &	B\big( \mathbf{x}^*, \mathbf{y}^* \big) = v j_{\mathrm{ave}} + \beta
 \end{align}
 for the transmitter and jammer, respectively.
\end{lemma}
\begin{proof}
 Follows from conditions (I.5) and (II.5) in Table~\ref{Table_KKT_Conditions_for_Bimatrix_Game}.
\end{proof}
% 
% 
% ****************************************************************************************
% 
\section{Connection to Quadratic Programing}
\label{Sec_3_Quadratic_Programming}

While Theorem~\ref{Theorem_2} gives the necessary and sufficient conditions for the strategy profile $(\mathbf{x}^*, \mathbf{y}^*)$ to be a NE of $\mathcal{G}$, it does not provide a constructive way to find the NE solution(s) and the equilibrium pairs of the constrained bimatrix game.
We have previously shown \cite{Koorosh_2013} that for every constrained two-player zero-sum game there exists an equivalent linear program whose solution yields a NE for the game and every NE of the game is a solution of the corresponding linear program.

In this section, we show that there exist a similar connection between the NE solutions and equilibrium pairs of the constrained bimatrix games and global maximum(s) of a \emph{quadratic program}%
\footnote{The connection between standard bimatrix games and quadratic programs was first shown in \cite{Manag_1964}.}. 
In the following theorem, we show that the global maximum of the quadratic program in~\eqref{Eq_Quadratic_Program_Equivalent} subject to the constraints in~\eqref{Eq_Quadratic_Program_Constraints} satisfies all conditions of Theorem~\ref{Theorem_2} and therefore, the corresponding maximizer is a NE solution of~$\mathcal{G}$. 

% 
% *********************************
% *** THEOREM ***
% 
\begin{theorem}
\label{Theorem_Quadratic_Program}
 Let $\mathcal{G} = \big(A, B, \mathbf{x}, \mathbf{y}, r_{\mathrm{ave}}, j_{\mathrm{ave}}\big)$ be a constrained bimatrix game with $A, B \in \mathbb{R}^{m\times n}$. The strategy pair $\big(\mathbf{x}^*, \mathbf{y}^* \big)$ is a Nash equilibrium of $\mathcal{G}$ if and only if there exist scalers $u^*,v^*\geq0$ and $\alpha^*,\beta^* \in \mathbb{R}$ such that 
 $\big(\mathbf{x}^*, \mathbf{y}^*, u^*,v^*,$ $ \alpha^*,\beta^* \big)$ is a global maximizer of the following quadratic program
\begin{equation}
 \label{Eq_Quadratic_Program_Equivalent}
 \underset{\mathbf{x}, \mathbf{y}, u,v, \alpha, \beta }{\mathrm{maximize}} \ \mathbf{x}^T \big(A + B \big) \ \mathbf{y} - 
 ur_{\mathrm{ave}} - vj_{\mathrm{ave}} - \alpha - \beta
\end{equation}
{subject to:}
\begin{equation}
 \label{Eq_Quadratic_Program_Constraints}
 \left\{
 \begin{array}{lr}
  A\mathbf{y} - u \mathbf{r}_{m\times1} - \alpha \mathbf{1}_{m\times1} \leq \mathbf{0} 			& (\ref{Eq_Quadratic_Program_Constraints}.1)\\ 
  \mathbf{x}^T B - v \mathbf{j}_{n\times1} - \beta \mathbf{1}_{n\times1} \leq \mathbf{0} 		& (\ref{Eq_Quadratic_Program_Constraints}.2)\\ 
  \mathbf{r}^T \mathbf{x} - r_{\mathrm{ave}} \leq 0							& (\ref{Eq_Quadratic_Program_Constraints}.3)\\
  \mathbf{j}^T \mathbf{y} - j_{\mathrm{ave}} \leq 0							& (\ref{Eq_Quadratic_Program_Constraints}.4)\\ 
  \mathbf{1}^T \mathbf{x} - 1 = 0									& (\ref{Eq_Quadratic_Program_Constraints}.5)\\
  \mathbf{1}^T \mathbf{y} - 1 = 0									& (\ref{Eq_Quadratic_Program_Constraints}.6)\\
  -\mathbf{x}, -\mathbf{y}\leq \mathbf{0} ,\  -u, -v \leq 0 \ \mathrm{and}\ \alpha, \beta \in \mathbb{R}& (\ref{Eq_Quadratic_Program_Constraints}.7)\\
%   -\mathbf{y}\leq \mathbf{0}\\
%   -u\leq 0\\ 
%   -v\leq 0\\
 \end{array}\right.
\end{equation}
\end{theorem}
% 
% *********************************
% 
% 
\begin{proof}
	First, notice that the constraints in \eqref{Eq_Quadratic_Program_Constraints} satisfy all the conditions of Table~\ref{Table_KKT_Conditions_for_Bimatrix_Game} except for (I.5), (I.6) and (II.5), (II.6). As a result, if we show that the global maximum of the quadratic program in \eqref{Eq_Quadratic_Program_Equivalent} satisfies these additional conditions, then, by Theorem~\ref{Theorem_2}, it must be a NE solution of $\mathcal{G}$.
	If we premultiply (\ref{Eq_Quadratic_Program_Constraints}.1) by $\mathbf{x}^T$ and use (\ref{Eq_Quadratic_Program_Constraints}.3) to simplify the result we have
% 	
% 	\begin{equation}
% 	\mathbf{x}^T A\ \mathbf{y} - u\mathbf{x}^T \mathbf{r} - \alpha  \leq 0
% 	\end{equation}
% 	
% 	since $u\geq 0$ and from \eqref{Eq_Quadratic_Program_Constraints}-(3) we obtain the following inequality
% 	
	\begin{equation}
	  \label{Eq_Z01}
	  \mathbf{x}^T A\ \mathbf{y} - ur_{\mathrm{ave}} - \alpha  \leq 0
	\end{equation}
	since $\mathbf{x}^T$  is a probability vector and $u\geq 0$.
	Similarly, we can obtain the following inequality from (\ref{Eq_Quadratic_Program_Constraints}.2) and (\ref{Eq_Quadratic_Program_Constraints}.4):
	\begin{equation}
	  \label{Eq_Z02} 
	  \mathbf{x}^T B\ \mathbf{y} - vj_{\mathrm{ave}} - \beta  \leq 0
	\end{equation}
	by combining  inequalities \eqref{Eq_Z01} and \eqref{Eq_Z02} we observe that
	\begin{equation*}
	  f\big( \mathbf{x}, \mathbf{y}, u,v, \alpha,\beta \big) \triangleq \mathbf{x}^T \big(A + B \big) \ \mathbf{y} - 
	  ur_{\mathrm{ave}} - vj_{\mathrm{ave}} - \alpha - \beta \leq 0
	\end{equation*}
	Thus, any set of variables $\big(\mathbf{x}^*, \mathbf{y}^*, u^*,v^*,$ $ \alpha^*,\beta^* \big)$ that satisfies
	\begin{equation}
	\label{Eq_f_Function_Definition}
	  f\big(\mathbf{x}^*, \mathbf{y}^*, u^*,v^*, \alpha^*,\beta^* \big) =0
	\end{equation}
	is a {\emph{global}} maximum of \eqref{Eq_Quadratic_Program_Equivalent}. Next, we will consider the KKT necessary conditions for optimality for the optimization problem in \eqref{Eq_Quadratic_Program_Equivalent}. To find the necessary KKT conditions, we stack the variables in the following vector and we take the gradients in the same order.
	\begin{equation}
	\mathbf{z}^T \triangleq
	\begin{bmatrix}
	  \mathbf{x}^T_{1\times m}&
	  \mathbf{y}^T_{1\times n}&
	  u&
	  v&
	  \alpha&
	  \beta&
	\end{bmatrix}^T_{(m+n+4)\times1}
	\end{equation}
	From this point forward, we assume all variables are  optimal and for convenience, we drop the $^*$ from the variables. Primal feasibility conditions are identical to the constraints in \eqref{Eq_Quadratic_Program_Constraints}. The dual feasibility condition necessitates%
	\footnote{The KKT conditions are the necessary conditions (not sufficient)  since the objective function in \eqref{Eq_Quadratic_Program_Equivalent} is non-convex.}
	that, for the global maximizer of \eqref{Eq_f_Function_Definition}, $\mathbf{z}$, the gradient of $f(\mathbf{z})$ must be a linear combination of the gradients of the binding constraints in \eqref{Eq_Quadratic_Program_Constraints}, i.e., we must have
	\begin{multline}
	\label{Eq_KKT_Conditions_for_Quadratic_Program}
	\begin{bmatrix}
	  (A+B)\mathbf{y}\\
	  (A+B)^T \mathbf{x}\\
	  -r_{\mathrm{ave}}\\
	  -j_{\mathrm{ave}}\\
	  -1\\
	  -1\\
	\end{bmatrix}
	- \sum_{i=1}^{m}\lambda_i
	\begin{bmatrix}
	  \mathbf{0}\\
	  A_{i,:}^T \\
	  -r_i \\
	  0\\
	  -1\\
	  0
	\end{bmatrix}
	- \sum_{k=1}^{n} \mu_k
	  \begin{bmatrix}
	  B_{:,k} \\
	  \mathbf{0}\\
	  0 \\
	  -j_k\\
	  0\\
	  -1
	\end{bmatrix}
	- b_1 
	\begin{bmatrix}
	  \mathbf{r} \\
	  \mathbf{0} \\
	  0 \\
	  0 \\
	  0 \\
	  0 \\
	\end{bmatrix}\\
% 	\end{equation*}
	%  
% 	\begin{multline}
% 	\label{Eq_KKT_Conditions_for_Quadratic_Program}
	-b_2 
	\begin{bmatrix}
	  \mathbf{0} \\
	  \mathbf{j} \\
	  0 \\
	  0 \\
	  0 \\
	  0 \\
	\end{bmatrix}
	- a_1
	\begin{bmatrix}
	  \mathbf{1} \\
	  \mathbf{0} \\
	  0 \\
	  0 \\
	  0 \\
	  0 \\
	\end{bmatrix}
	-a_2
	  \begin{bmatrix}
	  \mathbf{0} \\
	  \mathbf{1} \\
	  0 \\
	  0 \\
	  0 \\
	  0 \\
	\end{bmatrix}
	  -\sum_{i=1}^{m} \phi_i
	  \begin{bmatrix}
	  \mathbf{e}_i \\
	  \mathbf{0} \\
	  0 \\
	  0 \\
	  0 \\
	  0 \\
	\end{bmatrix}
	  -\sum_{k=1}^{n} \theta_k
	  \begin{bmatrix}
	  \mathbf{0} \\
	  \mathbf{e}_k \\
	  0 \\
	  0 \\
	  0 \\
	  0 \\
	\end{bmatrix}\\
	- \sigma_1
	\begin{bmatrix}
	  \mathbf{0} \\
	  \mathbf{0} \\
	  -1 \\
	  0 \\
	  0 \\
	  0 \\
	\end{bmatrix}
	  -\sigma_2 
	  \begin{bmatrix}
	  \mathbf{0} \\
	  \mathbf{0} \\
	  0 \\
	  -1 \\
	  0 \\
	  0 \\
	\end{bmatrix} = %\mathbf{0}_{(m+n+4)\times 1}
		  \begin{bmatrix}
	  \mathbf{0}_{m\times 1} \\
	  \mathbf{0}_{n\times 1} \\
	  0_{1\times 1} \\
	  0_{1\times 1} \\
	  0_{1\times 1} \\
	  0_{1\times 1} \\
	\end{bmatrix} \quad
	\begin{array}{c}
		\text{(i)}\\
		\text{(ii)}\\
		\text{(iii)}\\
		\text{(iv)}\\
		\text{(v)}\\
		\text{(vi)}\\
	\end{array}
\end{multline}
where we have taken the gradients of the constraints in the same order as in \eqref{Eq_Quadratic_Program_Constraints}, and $A_{i,:}$ and $\mathbf{e}_i$ denote the $i$'th row of $A$ and the $i$'th basis vector, respectively.  
Additionally, the KKT multipliers must satisfy
\begin{equation}
\left\{
\begin{aligned}
 & \lambda_i, \ \phi_i \geq 0  	&\quad \text{for}\ i=1,\dots,m\\
 & \mu_k,\ \theta_k \geq 0	&\quad \text{for}\ k=1,\dots,n\\
 & b_1,\ b_2,\ \sigma_1,\ \sigma_2 \geq 0 &\\
 & a_1,\ a_2 \in \mathbb{R} &\\
\end{aligned}
\right.
\end{equation}

{\color{black} By inspecting parts (v) and (vi) of the systems of vector equations in \eqref{Eq_KKT_Conditions_for_Quadratic_Program},  we observe that for the KKT multipliers $\lambda_i$ and $\mu_k$ we have}
\begin{equation}
\sum_{i=1}^{m} \lambda_i = 1, \quad \lambda_i \geq 0 
\quad \text{and} \quad
\sum_{k=1}^{n} \mu_k = 1, \quad \mu_k \geq 0 
\end{equation}
Now, let
\begin{equation}
\label{Eq_KKT_Multipliers_Substitutions}
 \begin{aligned}
  & \lambda_i = x_i	\quad \text{for}\ i = 1, \dots, m \\
  & \mu_k = y_k		\quad \text{for}\ k = 1, \dots, n \\
  & a_1 = \alpha,\  a_2 = \beta \\
 \end{aligned}
\end{equation}
and note that, because of the constraints on the KKT multipliers $\lambda_i$, $\mu_k$, $a_1$ and $a_2$, we are allowed to make these assumptions.
From parts (iii) and (iv) of \eqref{Eq_Quadratic_Program_Constraints} we obtain
\begin{equation}
\label{Eq_Quadratic_Program_Dual_Part_III_and_IV}
% \boxed{
 \begin{aligned}
  &	\mathbf{j}^T \mathbf{y} - j_{\mathrm{ave}} + \sigma_2 = 0 \qquad \sigma_2 \geq 0 \\
  &	\mathbf{x}^T\mathbf{r} - r_{\mathrm{ave}} + \sigma_1 = 0 \qquad  \sigma_1 \geq 0 \\
  \end{aligned}
%   }
\end{equation}
and finally, from parts (i) and (ii) we have
\begin{equation}
\label{Eq_Quadratic_Program_Dual_Part_I_and_II}
 \begin{aligned}
  &	\mathbf{x}^T B - b_2 \mathbf{j}^T - \beta \mathbf{1}^T + \boldsymbol{\theta}^T = 0 
   	                   & \qquad b_2, \boldsymbol{\theta} \geq 0	\\
  &	A \mathbf{y} - b_1 \mathbf{r} - \alpha \mathbf{1} + \boldsymbol{\phi} = 0
			  & \qquad b_1, \boldsymbol{\phi} \geq 0		\\    
 \end{aligned}
\end{equation}
By substituting the KKT multipliers with the variables given in \eqref{Eq_KKT_Multipliers_Substitutions} we may write the complementary slackness conditions for \eqref{Eq_Quadratic_Program_Constraints} as follows
\begin{equation}\label{Eq_Quadratic_Program_Complementary_1}
%  (1): \boldsymbol{\lambda}^T \big( A \mathbf{y} - u\mathbf{r} - \alpha \mathbf{1} \big)
%       \overset{(\boldsymbol{\lambda} = \mathbf{x})}{=} 
%       \boxed{
      \mathbf{x}^T  A\ \mathbf{y} - u \mathbf{x}^T \mathbf{r} - \alpha = 0
%       }
\end{equation}
\begin{equation}
\label{Eq_Quadratic_Program_Complementary_2}
%  (2): \big( \mathbf{x}^T B - v\mathbf{j}^T - \beta \mathbf{1}^T \big) \boldsymbol{\mu}
%       \overset{(\boldsymbol{\mu} = \mathbf{y})}{=} 
%       \boxed{
      \mathbf{x}^T  B\ \mathbf{y} - v \mathbf{j}^T \mathbf{y} - \beta = 0
%       }
\end{equation}
\begin{equation}
\label{Eq_Quadratic_Program_Complementary_3}
%  (3): b_1 \big( \mathbf{r}^T \mathbf{x} - r_{\mathrm{ave}} \big) = 0 
%   \quad
%   \Rightarrow
%   \quad
%   \boxed{
  b_1 \mathbf{r}^T \mathbf{x} = b_1 r_{\mathrm{ave}}
%   }
\end{equation}
\begin{equation}
\label{Eq_Quadratic_Program_Complementary_4}
%  (4): b_2 \big( \mathbf{j}^T \mathbf{y} - j_{\mathrm{ave}} \big) = 0 
%   \quad
%   \Rightarrow
%   \quad
%   \boxed{
  b_2 \mathbf{j}^T \mathbf{y} = b_2 j_{\mathrm{ave}}
%   }
\end{equation}
and
\begin{equation}
\label{Eq_Quadratic_Program_Complementary_7_to_10}
% \boxed{
%  \begin{array}{cl}
   \boldsymbol{\phi}^T\mathbf{x} = 
   \boldsymbol{\theta}^T\mathbf{y} = 
   \sigma_1 u = 
   \sigma_2 v = 0 \\
%  \end{array}
%  }
\end{equation}

Now, if we multiply the first relation in \eqref{Eq_Quadratic_Program_Dual_Part_III_and_IV} by $v$ and use \eqref{Eq_Quadratic_Program_Complementary_7_to_10} to simplify the result we have
\begin{equation}
\begin{aligned}
  & v \big( \mathbf{j}^T \mathbf{y} - j_{\mathrm{ave}} \big) + v\sigma_2 = 0 \\
  & \qquad \Rightarrow %\boxed{\boxed{
			v \big( \mathbf{j}^T \mathbf{y} - j_{\mathrm{ave}} \big) = 0
% 			}}
\end{aligned}
\end{equation}
which is identical to condition (II.6) in Table~\ref{Table_KKT_Conditions_for_Bimatrix_Game}. Similarly, from the second relation in \eqref{Eq_Quadratic_Program_Dual_Part_III_and_IV}, we can obtain 
\begin{equation}
%  \boxed{\boxed{
 u \big( \mathbf{x}^T\mathbf{r} - r_{\mathrm{ave}} \big)= 0
%  }}
\end{equation}
which gives us condition (I.6) in Table~\ref{Table_KKT_Conditions_for_Bimatrix_Game}. Finally, if we post-multiply \eqref{Eq_Quadratic_Program_Dual_Part_I_and_II} by $\mathbf{y}$ and use \eqref{Eq_Quadratic_Program_Complementary_7_to_10} to simplify the result we have
\begin{align}
 & \mathbf{x}^T B\ \mathbf{y} - b_2 \mathbf{j}^T \mathbf{y} - \beta \mathbf{1}^T \mathbf{y} + \boldsymbol{\theta}^T\mathbf{y} = 0  \notag \\
 & \qquad \Rightarrow \qquad \mathbf{x}^T B\ \mathbf{y}- b_2 \mathbf{j}^T \mathbf{y} - \beta   = 0  \label{Eq_Z05}
\end{align}
comparing \eqref{Eq_Z05} with \eqref{Eq_Quadratic_Program_Complementary_2} we notice that $b_2 = v$ and by using \eqref{Eq_Quadratic_Program_Complementary_3} we obtain the desired result:% (Condition (II.5) in Table~\ref{Table_1}):
\begin{equation}
\label{Eq_Z03}
% \boxed{\boxed{
 \mathbf{x}^T B\ \mathbf{y}- v j_{\mathrm{ave}} - \beta   = 0
%  }}
\end{equation}
Similarly, we can show that $b_1 = u$ and 
\begin{equation}
\label{Eq_Z04}
% \boxed{\boxed{
 \mathbf{x}^T A \ \mathbf{y} - u r_{\mathrm{ave}} - \alpha = 0
%  }}
\end{equation}
Conditions \eqref{Eq_Z03} and \eqref{Eq_Z04} are exactly conditions (II.5) and (I.5) in Table~\ref{Table_KKT_Conditions_for_Bimatrix_Game} and as a result, the maximizer of the quadratic program in \eqref{Eq_Quadratic_Program_Equivalent}, subject to constraints in \eqref{Eq_Quadratic_Program_Constraints}, satisfies all the conditions of Theorem~\ref{Theorem_2} and, hence, is a Nash equilibrium of $\mathcal{G}$. The last step is to show that the set of variables $\big( \mathbf{x}, \mathbf{y}, u,v, \alpha,\beta \big)$ is indeed a global maximizer of \eqref{Eq_Quadratic_Program_Equivalent}. Adding \eqref{Eq_Z03} to \eqref{Eq_Z04} gives us the desired result.
\begin{equation}
\mathbf{x}^T \big(A + B \big) \ \mathbf{y} - ur_{\mathrm{ave}} - vj_{\mathrm{ave}} - \alpha - \beta = 0
\end{equation}

The converse of theorem states that if $(\mathbf{x}^*, \mathbf{y}^*)$ is a NE pair, then, $\big( \mathbf{x}^*, \mathbf{y}^*, u^*,v^*, \alpha^*,\beta^*\big)$ is a global maximizer of \eqref{Eq_Quadratic_Program_Equivalent}. By using Lemma \ref{Lemma_Bimatrix_Game_Expected_Payoffs} and the necessary and sufficient conditions in Table~\ref{Table_KKT_Conditions_for_Bimatrix_Game} it can be easily verified that  
\begin{equation}
 {\mathbf{x}^*}^T \big(A + B \big) \ \mathbf{y}^* - u^*r_{\mathrm{ave}} - v^*j_{\mathrm{ave}} - \alpha^* - \beta^* = 0
\end{equation}
and hence the NE solution of the constrained bimatrix game $\mathcal{G} = \big( A,B, $ $\mathbf{r}, \mathbf{j}, r_{\mathrm{ave}}, j_{\mathrm{ave}} \big)$ is indeed a global maximum of the quadratic program defined in \eqref{Eq_Quadratic_Program_Equivalent}. This concludes the proof.
\end{proof}

% 
% ****************************************************************************************
% 
\section{A Special Case: Packetized AWGN Link under \\ Power Limited Jamming}
\label{Sec_4_Case_Study}
\begin{figure}%[t]
\centering
 \includegraphics[width = 1\linewidth]{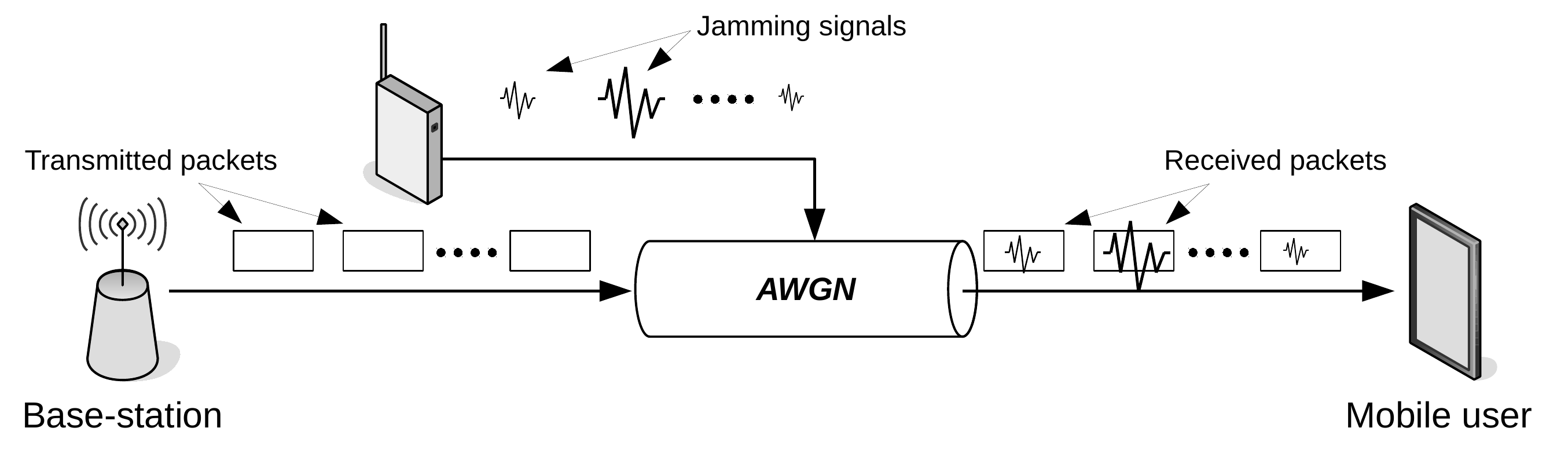}
 \caption{Packetized AWGN link under power limited jamming.}
 \label{Fig_System_Model}
\end{figure}
%
%***************************************************************************************************
%

In this section we use the framework we developed in the previous sections to study a typical jamming problem and show that the constrained bimatrix game can be used to formulate this typical problem.

Consider the wireless communication system shown in Figure~\ref{Fig_System_Model}. The communication link between a base-station (transmitter) and a mobile user (receiver) is a single-hop, packet-switched, AWGN channel with fixed and known noise variance, $N$, measured at the receiver's side. Furthermore, assume the communication link is being disrupted by an average power limited additive Gaussian jammer with flat power spectral density. 
The impact of the Gaussian jammer on the communication link is the reduction of the effective signal to noise ratio (SNR) at the receiver from ${P_T}/{N}$ to ${P_T}/(N+J)$, where $J$ represents the jammer power (variance) and $P_T$ is the transmitter power, both measured at the receiver side.

We assume that the jammer uses a set of discrete jamming power levels denoted by $\mathcal{J}$. The jammer may use any jamming power but must maintain an
overall average power constraint, denoted by $J_{\text{ave}}$. The jammer uses  his available power levels according to a probability distribution (his strategy) and his goal is to  cause the maximum damage to the communication link by destroying as many packets as possible while maintaining the average power constraint.

The base-station has a rate adaptation block with $n$ {different} but {fixed} rates. 
Transmission rates are bounded between minimum and maximum rates denoted by $R_{\min}$ and $R_{\max}$, respectively.
Without loss of generality, we assume the rates are sorted in a decreasing order. Hence, the base-station's action set, denoted by $\mathcal{R}$, becomes
\begin{equation}\label{Eq:Trans_Action_Set_Arbitrary_Rates}
 \mathcal{R} = \big\{ { R_0=  R_{\max} > \cdot \cdot> R_i > \cdot \cdot > R_{n-1} = R_{\min}} \big\}_{\text{(nats/trans)}}
\end{equation} 
Assuming $R_{\max}$ is feasible and packets are long enough that channel capacity theorem could be applied to each packet,
it follows from the capacity of the discrete-time AWGN channel that we must have
\begin{equation}\label{Eq:Trans_Minimum_Power}
 P_T \geq P_{\min} = N \left( e^{2R_{\max}} -1 \right) \\
\end{equation}
to make all transmission rates viable. Throughout the rest of this section, we assume that the base-station transmits data packets at a fixed and known power, $P_T$, where $P_T\geq P_{\min}$. The base-station uses the available rates according to a probability distribution (his strategy) and his goal is to find an optimal strategy to maximize the average throughput of the channel subject to jamming.

Given that the channel noise variance is fixed and known, corresponding to each transmission rate $R_j \in \mathcal{R} $ there exists a certain jammer power, $\widehat{J_j} \geq 0$, 
such that if the actual jamming power used by the jammer is less than $\widehat{J_j}$, then 
% below which 
reliable communication is possible, i.e.,
\begin{equation}\label{Eq:Jamming_Powers_Corresponding_Rates}
 \begin{aligned}
  & R_j = \frac{1}{2} \log \left( 1 + \frac{P_T}{N+ \widehat{J_j}} \right)  \\ %\ \ j = 1,\cdots, n \\
%   & \Rightarrow \qquad \widehat{J_j} = \frac{P_T}{e^{2R_j} - 1} - N\\
  & \Rightarrow \qquad \widehat{J_j} = \frac{P_T}{e^{2R_j} - 1} - N \quad  j = 0,\cdots,  n-1 \\
 \end{aligned}
\end{equation}

Assuming that $\mathcal{R}$ is publicly available (such as the typical rates of IEEE~802.11 standard) and $P_T$ and $N$ could be estimated, the jammer can use~\eqref{Eq:Jamming_Powers_Corresponding_Rates} to construct his action set, specifically, consider the following action set
\begin{equation}
  \label{Eq:Jammer_Action_Set}
  \mathcal{J} = \Big\{ J_0, J_1, \cdots, J_j, \cdots, J_{n} \Big\}
\end{equation}
where $J_j$ for $0\leq j \leq n$ is given by
\begin{equation*}%\label{Eq:Jamming_Powers_Corresponding_Rates_Plus_Delta}
  J_j = 
  \begin{cases}
    0 & j = 0\\
    \widehat{J}_{j-1} + \delta N = \frac{P_T}{e^{2R_{j-1}} - 1} + (\delta - 1)N & j = 1,\dots, n\\
  \end{cases}
\end{equation*}
The jammer adds $ \delta N $ with $\delta > 0 $ to his non-zero jamming powers to make sure that $R_{j}$ is greater than the channel capacity corresponding to $J_j$. Jammer's mixed-strategy set, $\mathbf{Y}^{n+1}_{J_{\text{ave}}}$, is then the set of all probability vectors that result in an average power less than or equal to $J_{\text{ave}}$, i.e.,
\begin{equation}\label{Eq:Jammer_Mixed_Set_LE}
 \mathbf{Y}^{n+1}_{J_{\text{ave}}} = \{ \boldsymbol{y}_{(n+1)\times 1} \in \mathbf{Y}^{n+1} \big| \; \boldsymbol{y}^T 
 \boldsymbol{J} \leq J_{\text{ave}} \}
\end{equation}
where $\boldsymbol{y}_{(n+1)\times 1}$ and $\boldsymbol{J}_{(n+1) \times 1}$ are jammer's mixed-strategy and jamming power vectors, respectively, and $\mathbf{Y}^{n+1}$ is a standard $(n+1)$-simplex.

Since destroyed packets do not contribute to the average throughput of the communication system, 
the payoff per transmitted packet, $C \big( R_i , J_j \big)$, for the pure-strategy pair $ \big( R_j , J_j \big)\in \mathcal{R} \times \mathcal{J}$ is equal to the transmission rate of that packet if the packet is recovered, and zero if it is destroyed, i.e.,
\begin{equation}\label{Eq_Utility_Function_Rate}
 C \big( R_i , J_j \big)_{\text{nats/trans.}} = 
  \begin{cases}
   R_i & j < i		\\
   0   & j \geq i	\\
  \end{cases}
  \quad (R_i,\ J_j)  \in \mathcal{R} \times \mathcal{J}
\end{equation}
% 

%
%*************************************************************************
%
\begin{figure}
\centering
  \includegraphics[width = 1.0\linewidth]{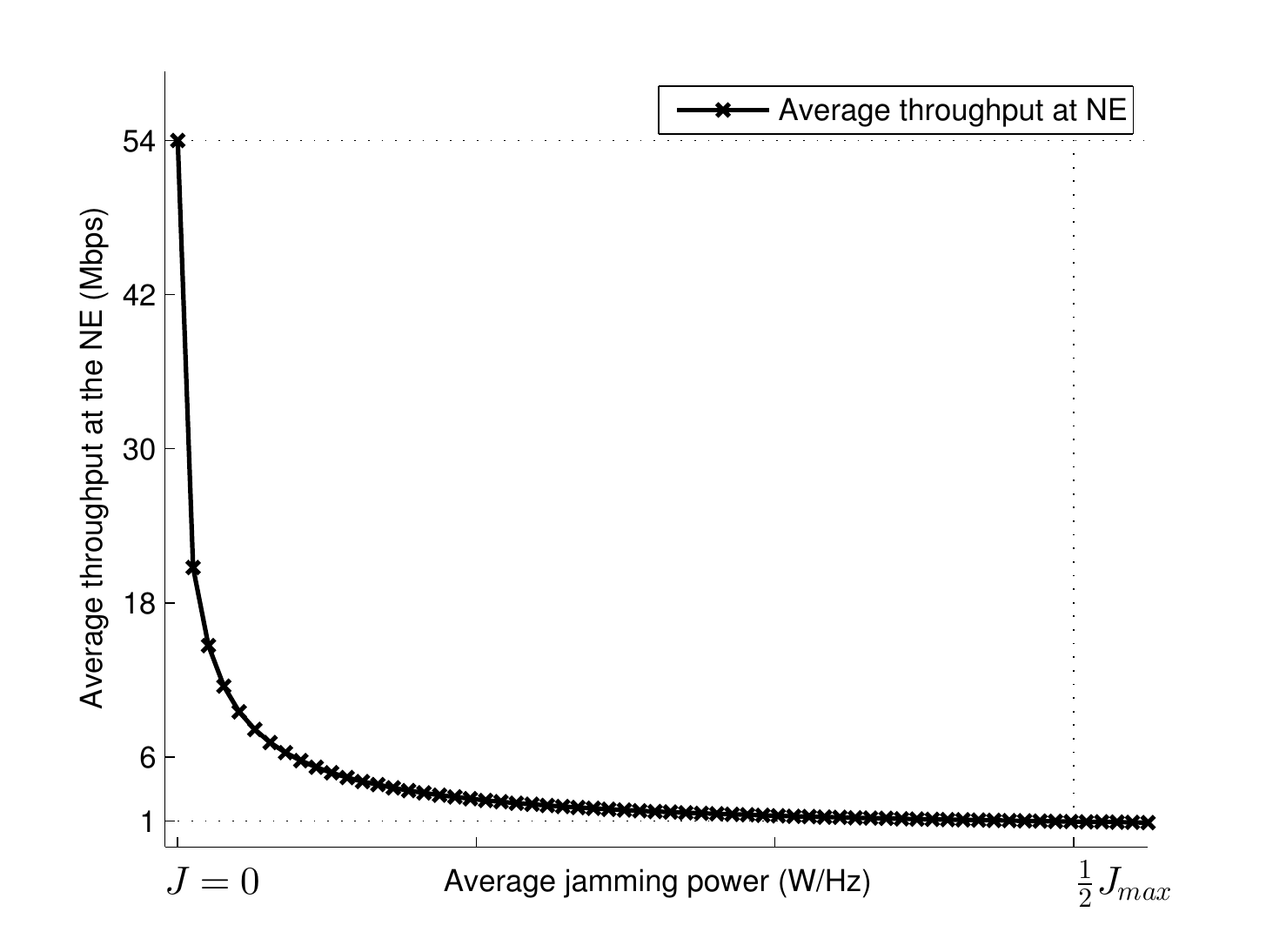}  
  \caption{Average throughout at the NE as a function of $J_{\text{ave}}$.}
  \label{Fig_2}
\end{figure}
%
%*************************************************************************
% 

Therefore the payoff matrix corresponding to  \eqref{Eq_Utility_Function_Rate}, where the base-station is the row player, will be an $n\times (n+1)$ matrix with zero elements above the main diagonal, i.e.,
\begin{equation}
\label{Eq_Payoff_Matrix_Trans}
  C = 
  \left[
  \renewcommand{\arraystretch}{.7}
  \renewcommand{\arraycolsep}{2pt}
 \begin{array}[m]{cccccc}
  R_0     & 0      &        & \cdots &         & 0      \\
  \vdots  &        & \ddots &        &         & \vdots \\ 
  R_i     & \cdots & R_i    & 0      & \cdots  & 0      \\
  \vdots  &        &        &        & \ddots  & \vdots \\
  R_{n-1} &        &        & \cdots & R_{n-1} & 0      \\
 \end{array}
 \right]_{n\times (n+1)} 
\end{equation}
Let $\boldsymbol{x}_{n\times 1}$ and $\boldsymbol{y}_{(n+1)\times 1 }$ be the base-station's and jammer's mixed-strategies, respectively. Then the base-station's problem becomes the following maximization problem
\begin{equation}
 \label{Eq_Transmitter_Probelm}
 \underset{\boldsymbol{x} \in \mathbf{X}^n}{\text{maximize}} \quad \boldsymbol{x}^T C \boldsymbol{y} \qquad \text{for all}\  \boldsymbol{y} \in \mathbf{Y}^{n+1}_{J_{\text{ave}}}
\end{equation}
It can be proved \cite{Koorosh_2013} that this problem has closed form solution and the average throughput at the Nash equilibrium as a function of $J_{\text{ave}}$ is given by
\begin{equation}
 \label{Eq_Average_Throughput_at_NE}
 C\big( \boldsymbol{x}^*, \boldsymbol{y}^* \big) = 
  \frac{J_{m+1} - J_{\text{ave}}}{J_{m+1} - J_{\text{ave},m}} R_m 
    \quad 
    J_{\text{ave},m} \leq J_{\text{ave}} < J_{\text{ave},m+1}
\end{equation}
for $1\leq m < n-1$ and $J_{\text{ave},m}$ is defined as 
\begin{equation}
\label{Eq_J_ave_m}
 J_{\text{ave},m} = R_m \sum^{m}_{j=1} \left( R_j^{-1} - R_{j-1}^{-1} \right) J_{j} \qquad  1\leq m \leq n
\end{equation}
Figure~\ref{Fig_2} shows the average throughput of the communication link at the NE as a function of jammer's average power for a typical case. 
For this example, we use the range of rates from the IEEE~802.11 standard, i.e., we assume coded data rates of the base station are distributed between $R_{\min} = 1$~Mbps and $R_{\max} = 54$~Mbps and the channel bandwidth is $22$~MHz.

Since  jammer's goal is to maximize the number of destroyed packets, we define the jammer's payoff per packet to be $1$ if the packet is destroyed and $0$ if the packet is recovered. Thus, the jammer's utility function for the pure-strategy pair $(R_i,\ J_j)$ becomes
\begin{equation}
 \label{Eq_Utility_Function_Jammer}
 J \big( R_i , J_j \big) = 
  \begin{cases}
   0 & j < i		\\
   1   & j \geq i	\\
  \end{cases}
  \qquad (R_i,\ J_j)  \in \mathcal{R} \times \mathcal{J}
\end{equation}
and the payoff matrix corresponding to \eqref{Eq_Utility_Function_Jammer} becomes
\begin{equation}
\label{Eq_Payoff_Matrix_Jammer}
  J^T = 
  \left[
  \renewcommand{\arraystretch}{.7}
  \renewcommand{\arraycolsep}{2pt}
 \begin{array}{cccccc}
  0      & 1      &        & \cdots &        & 1 \\
  \vdots &        & \ddots &        &        & \vdots \\ 
  0      & \cdots & 0      & 1      & \cdots & 1 \\
  \vdots &        &        &        & \ddots & \vdots \\
  0      &        &        & \cdots & 0      & 1 \\
 \end{array}
 \right]_{n\times (n+1)}
\end{equation}

Comparison of the payoff matrices in \eqref{Eq_Payoff_Matrix_Trans} and \eqref{Eq_Payoff_Matrix_Jammer} clearly shows  base-station's and  jammer's conflicting goals; while the base-station's non-zero payoffs appear on or below the main diagonal of his payoff matrix ($C$), jammer's non-zero payoffs are above the main diagonal of his respective payoff matrix ($J^T$).
But in contrast to the zero-sum games, the sum of the two matrices in \eqref{Eq_Payoff_Matrix_Trans} and \eqref{Eq_Payoff_Matrix_Jammer} is not zero.

Since the jammer's utility function is not the negative of the base-station's utility function the jammer can play two different games to cause damage to the performance of the communication link. The jammer can simply ignore base-station's utility function and maximize his average utility based on his own payoff matrix. This game is equivalent to a constrained zero-sum game with matrix $J$ given in \eqref{Eq_Payoff_Matrix_Jammer} and average power constraint $J_{\text{ave}}$ where the jammer is the row player (maximizer).

It can be easily verified that any row in jammer's payoff matrix ($J$) is dominated by the last row which corresponds to his maximum jamming power ($J_{n}$). But because of the average jamming power constraint, $J_{\text{ave}}$, the jammer cannot use $J_{n}$ all the time. As a result the optimal strategy for the jammer is to use his maximum jamming power with probability $p = {J_{\text{ave}}}/{J_n}$ and not jam a packet with probability $(1-p)$. Therefore, jammer's expected payoff (average destroyed packets) as a function of his average power for the constrained zero-sum game becomes
\begin{equation}
\label{Eq_Payoff_Jammer_Zero_Sum}
 J^*_{\text{zero-sum}} \big( J_{\text{ave}} \big) = \frac{1}{J_n} J_{\text{ave}} \qquad 0\leq J_{\text{ave}} \leq J_n
\end{equation}
The optimal strategy for this zero-sum game (this strategy is called the jammer's \emph{maxmin} strategy) guarantees the payoff given in \eqref{Eq_Payoff_Jammer_Zero_Sum} regardless of the base-station's strategy.

An alternative approach for the jammer is to play the constrained bimatrix game $\mathcal{G} = \big( C, \ J^T,\ \boldsymbol{R},\ R_{\mathrm{ave}}, \ \boldsymbol{J},\ J_{\text{ave}} \big)$, where $\boldsymbol{R}$ is the base station rate vector. Since in this special case the base station does not have an average constraint on its strategies, $R_{\mathrm{ave}}$ is an arbitrary number that satisfies $R_{\mathrm{ave}} > \max R_i$. With this assumption, the condition (I.2) in Table~\ref{Table_KKT_Conditions_for_Bimatrix_Game} becomes redundant and from condition (I.5) in Table~\ref{Table_KKT_Conditions_for_Bimatrix_Game} it follows that $u=0$, hence, the quadratic program in~\eqref{Eq_Quadratic_Program_Equivalent} simplifies to
\begin{equation}
\label{Eq_Jammer_Problem_Bimatrix}
 \underset{\boldsymbol{x}, \boldsymbol{y},v,\alpha, \beta}{\text{maximize}} \quad \boldsymbol{x}^T \big( C + J^T \big) \boldsymbol{y} -vJ_{\text{ave}} -\alpha -\beta
\end{equation}
subject to
\begin{equation}
\left\{
 \begin{array}{l}
  C\boldsymbol{y} - \alpha \boldsymbol{1} \leq \boldsymbol{0} \\
  \boldsymbol{x}J^T - v\boldsymbol{J} - \beta \boldsymbol{1} \leq \boldsymbol{0} \\
%   \\
  \boldsymbol{J}^T \boldsymbol{y} - J_{\text{ave}} \leq 0 \\
%   \\
  \boldsymbol{1}^T \boldsymbol{y} - 1 = 0 \\
  \boldsymbol{1}^T \boldsymbol{x} - 1 = 0 \\
%   \\
  \boldsymbol{x},\ \boldsymbol{y},\ v \geq 0 \\
 \end{array}
 \right.
\end{equation}
and the expected payoff of the jammer at the NE becomes
\begin{equation}
 J^*_{\text{bimatrix}}\big( J_{\mathrm{ave}} \big) = v^*J_{\mathrm{ave}} + \beta^*
\end{equation}
where $v^*$ and $\beta^*$ are the global maximizers of \eqref{Eq_Jammer_Problem_Bimatrix}. 

It can be shown (see Theorem~\ref{Theorem_Case_Study}) that for certain values of $J_{\mathrm{ave}}$, the maximization problem in~\eqref{Eq_Jammer_Problem_Bimatrix} has a closed form solution. For these specific values, the expected payoff of the base station is equal to $R_m, m = 0, \dots, n-1$.

%
%*************************************************************************
% 
\begin{figure}[t]
\centering
  \includegraphics[width = 1.0\linewidth]{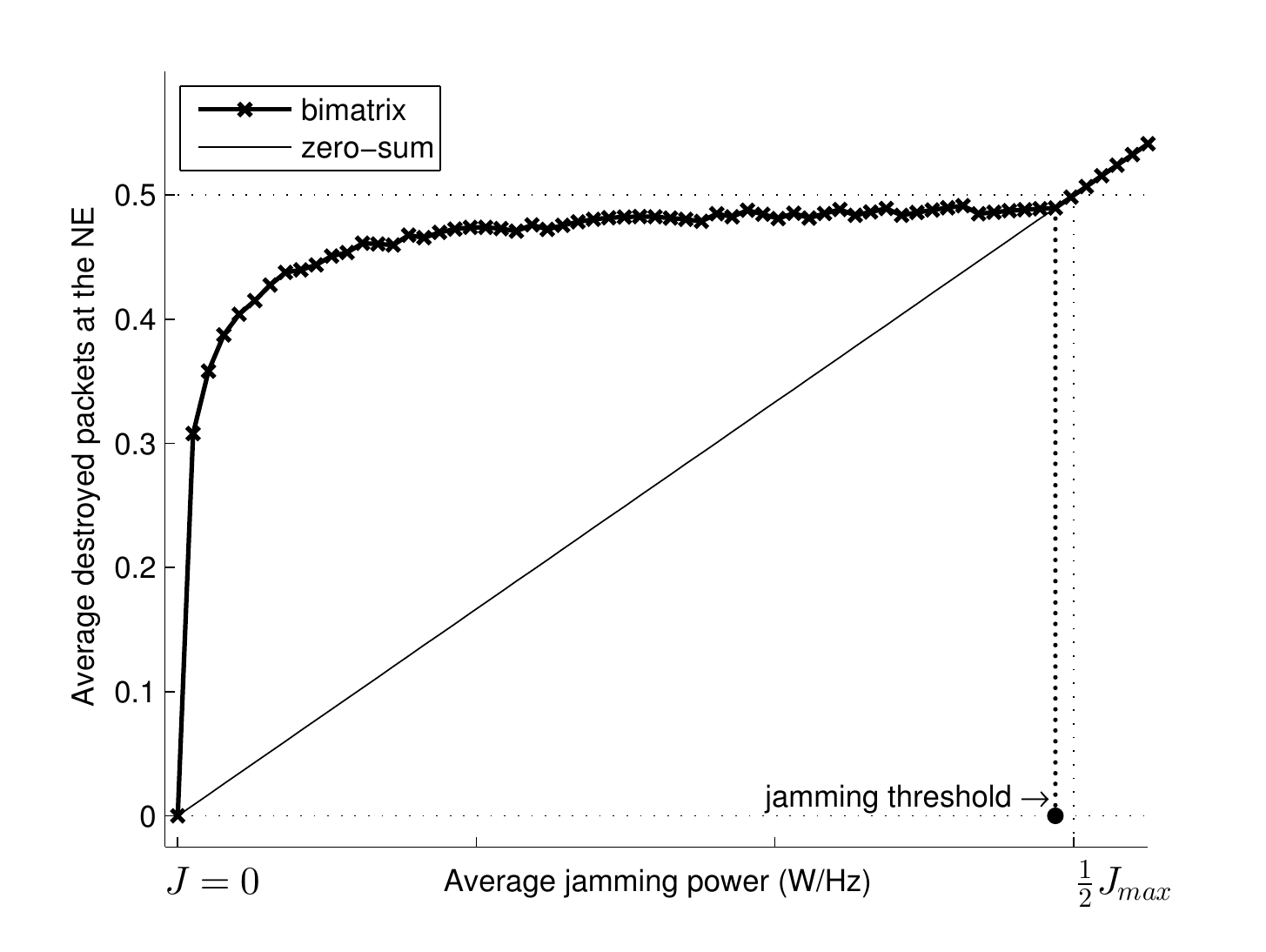} 
  \caption{Average destroyed packets at the NE as a function of $J_{\text{ave}}$.}
  \label{Fig_3}
\end{figure}
%
%*************************************************************************
% 

\begin{theorem}
\label{Theorem_Case_Study}
	In the constrained bimatrix game $\mathcal{G} = (C,\ J^T,$ $ \mathbf{R}, R_{\mathrm{ave}},$ $ \boldsymbol{J},\ J_{\mathrm{ave}})$ let
	\begin{equation}
		J_{\mathrm{ave}} = R_m \sum_{i=0}^{m} (R_i^{-1} - R_{i-1}^{-1}) J_i \quad \text{for}\ m = 0, \dots, n-1
	\end{equation}
	then, equilibrium pair solution, $(\boldsymbol{x^*}, \boldsymbol{y^*}) $, and the optimal mixed-strategies are given by
	\begin{equation}
	\begin{aligned}
		& \boldsymbol{x^*}^T = [ x_0, \dots, x_i, \dots x_m, \boldsymbol{0}] \quad x_i = J_m^{-1} ( J_i - J_{i-1} ) \\
		& \boldsymbol{y^*}^T = [ y_0, \dots, y_i, \dots y_m, \boldsymbol{0}] \quad y_i = R_m (R_{i}^{-1} - R_{i-1}^{-1} ) 
	\end{aligned}
	\end{equation}
	for the base station and the jammer respectively (where we used $R_{-1}^{-1} = J_{-1} = 0$). Furthermore, the expected payoffs of $\mathcal{G}$ at the NE are
	\begin{equation}
	\begin{aligned}
		& \boldsymbol{x^*}^T C \ \boldsymbol{y^*} = R_m \\
		& \boldsymbol{y^*}^T J \ \boldsymbol{x^*} = R_m J_{m}^{-1} \sum_{i=0}^{m} (R_{i}^{-1} - R_{i-1}^{-1})J_i \\
	\end{aligned}
	\end{equation}
\end{theorem}

\begin{proof}
	It is sufficient to show that there exist $v > 0 $ and $\alpha, \beta \in \mathbb{R}$ for which~\eqref{Eq_Jammer_Problem_Bimatrix} is zero. Let $v = J_{m}^{-1}$, $\alpha = R_m$ and $\beta  = 0$, then 
	\begin{equation}
		\boldsymbol{x^*}^T \big( C + J^T \big) \boldsymbol{y^*} -J_{m}^{-1}J_{\text{ave}} -R_m = 0
	\end{equation}

\end{proof}

Analytical study and numerical simulations verify that the expected payoff at the NE for the constrained bimatrix game strictly outperforms the zero-sum game if the average jamming power is less than a \emph{jamming threshold}, $J_{\text{TH}}$. That is, the expected payoff of the bimatrix game satisfies
\begin{equation}
 J^*_{\text{bimatrix}}\big( J_{\mathrm{ave}} \big) > J^*_{\text{zero-sum}}\big( J_{\mathrm{ave}} \big) \ \ \text{for all} \ 0 < J_{\text{ave}} < J_{\text{TH}}
\end{equation}
and
\begin{equation}
 J^*_{\text{bimatrix}}\big( J_{\mathrm{ave}} \big) = J^*_{\text{zero-sum}}\big( J_{\mathrm{ave}} \big) \ \ \text{for all} \  J_{\text{ave}} \geq J_{\text{TH}}
\end{equation}
The jamming threshold, $J_{\text{TH}}$, is the minimum average jamming power required %to force a transmitter to operate at his lowest rate all the time 
to force a transmitter to operate at his lowest rate in a single-hop packetized wireless link%
\footnote{Theoretical analysis suggests that for a single-hop packetized communication link under power limited jamming, such a threshold always exists \cite{Koorosh_2013}, \cite{Koorosh_2012}. Experimental studies confirm the existence of such a threshold on jammer's average power~\cite{hanawal2014game}.}%
.  
It can be proved (see \cite{Koorosh_2013} Theorem~4) that the minimum average jamming power that can force the transmitter to use his lowest rate is given by
\begin{equation}
\label{Eq_Jamming_Threshold}
 J_{\text{TH}} = R_{n-1} \sum^{n-1}_{j=1} \big( R^{-1}_{j} - R^{-1}_{j-1} \big) J_{j}
\end{equation}

Figure~\ref{Fig_3} shows a comparison between the expected payoff of the zero-sum game and the constrained bimatrix game for a typical case. As expected, the average payoff of the bimatrix game at the NE strictly dominates the zero-sum game for $J_{\text{ave}} < J_{\text{TH}}$ and the expected payoffs converge for $J_{\text{ave}} \geq J_{\text{TH}}$, i.e., the bimatrix game simplifies to a zero-sum game.

% 
% ****************************************************************************************
% 
\section{Conclusion}
\label{Sec_5_Conclusion}
We developed a constrained bimatrix game framework that can be used to model many practical jamming problems in packetized wireless networks. In contrast to the standard bimatrix games, in constrained bimatrix games the players' strategies must satisfy some additional average conditions, consequently, not all strategies are feasible and the existence of the NE is not
guaranteed anymore. 
We provided the necessary and sufficient conditions under which the existence of the Nash equilibrium (NE) is guaranteed and showed that the equilibrium pairs and the Nash equilibrium solution of this constrained game corresponds to the global maximum of a quadratic program. Finally, we studied a typical packetized wireless link under power limited jamming and showed that the game theoretic analysis of this typical problem yields rather surprising results.
% 
% ****************************************************************************************
% 
% 
% ****************************************************************************************
% 
% \appendices
\appendix[Existence of the Nash Equilibrium for \\ the Constrained Bimatrix Game]
\label{Appendix_I}
% 
%  

% 
% *******************************************************
% 
\begin{figure}[t]
 \centering
 \includegraphics[width = 1\linewidth]{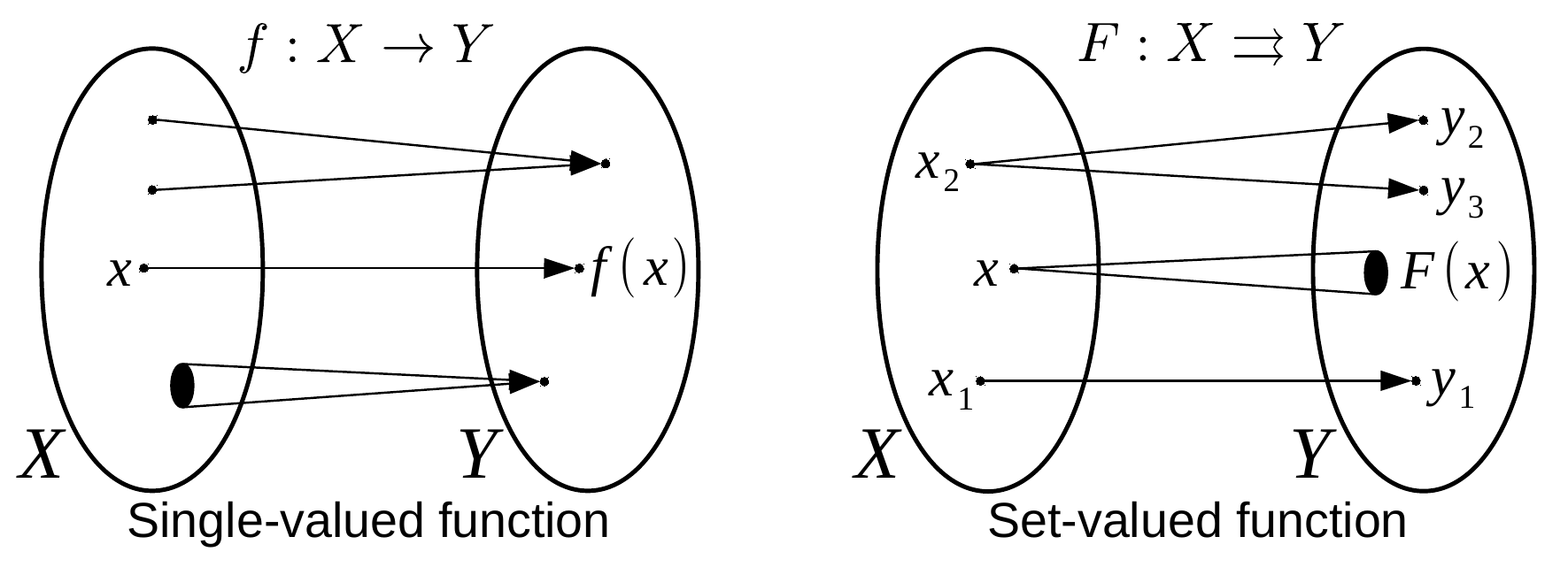}
 \caption{Single-valued function (left) vs. set-valued function (right).}
 \label{Fig_Set_Valued_Fuction}
\end{figure}
% 
% *******************************************************
% 

In game theory, fixed-point theorems are commonly used to prove that a model has an equilibrium point. In particular, Brouwer's fixed point theorem is often used to prove the existence of a solution for finite games (e.g., see~\cite{Nash51}), however, the approach used in~\cite{Nash51} (and similar approaches) cannot be extended to constrained game. 
As a consequence of the average constrains on mixed-strategies, arbitrary probabilities cannot be assigned to some pure-strategies and as a result, a more general approach is required.
In our approach, we use the \emph{Kakutani's fixed point theorem} to prove that our constrained bimatrix model has at least one equilibrium point. We start this section by providing some definitions.

  {\it Set-valued function} (set-function or correspondence): Denoted by $F: X \rightrightarrows Y$ is mapping from $X$ to non-empty subsets of $Y$, i.e., for all  $x\in X $ we have $ F(x) \in  2^Y -\ \emptyset$. As opposed to a \emph{single-valued function} (or simply, a function), a set-valued function can map its input to more than one output.  Figure~\ref{Fig_Set_Valued_Fuction} shows a comparison of a single-valued function and a set-valued function.

  {\it Convexed-valued function:} Let $F : X \rightrightarrows Y $ be a set-valued function then $F$ is convex-valued if $F(x)$ is a convex set for all $x\in X$.

  {\it Upper semi-continuous:} % \footnote{Same as continuity in functions.}
  $F$ is upper semi-continuous if the following holds: for every sequence $x_k $ in $ X$ that converges to some point $x \in X$ and for every sequence $y_k$ in $Y$ that converges to $y \in Y$, if $y_k \in F(x_k)$ for all $k\in\mathbb{N}$, then $y \in F(x)$.

  {\it Fixed point of a set-valued function:} Let $F: Z \rightrightarrows Z $ be a set-valued function then $x^* \in Z$ is a fixed point of $F$ if $x^* \in F(x^*)$.

The following theorem, known as Kakutani's fixed point theorem, provides the sufficient conditions for a set-valued function defined on a subset of Euclidean space to have a fixed point.

% 
%  *** Theorem ***
% 
\begin{theorem}[ Kakutani Fixed Point Theorem]
  Let $Z \subseteq \mathbb{R}^n$  be a nonempty compact and convex set and let $F\ : Z \rightrightarrows Z $ be an upper semi-continuous and convex-valued correspondence. Then F has a fixed point.
\end{theorem}
\begin{proof}
	See \cite{kakutani1941generalization}.
\end{proof}

% 
%  *** Theorem ***
% 
\begin{theorem}
  Let $\mathcal{G} = \big(A, B,$ $ \mathbf{r}, \mathbf{j}, r_{\mathrm{ave}}, j_{\mathrm{ave}}\big)$ be a constrained bimatrix game,
  then $\mathcal{G}$ has a Nash equilibrium solution if $r_{\mathrm{ave}}\geq \min r_i$ and $j_{\mathrm{ave}} \geq \min j_k$.
\end{theorem}

\begin{proof}
Let $F$ be a set-valued function defined on $(\widehat{\mathbf{X}} \times \widehat{\mathbf{Y}})$,
\begin{equation}
\label{Eq_F_Function}
 F:\ (\widehat{\mathbf{X}} \times \widehat{\mathbf{Y}}) \rightrightarrows (\widehat{\mathbf{X}} \times \widehat{\mathbf{Y}})
\end{equation}
where $\widehat{\mathbf{X}}$ and $\widehat{\mathbf{Y}}$ are defined in \eqref{Eq_Player_1_Constrained_Strategy_Set} and \eqref{Eq_Player_2_Constrained_Strategy_Set}, respectively.
Obviously, $\widehat{\mathbf{X}}\subset \mathbb{R}^m$ and $\widehat{\mathbf{Y}}\subset \mathbb{R}^n$ are non-empty, closed and   convex subsets. (Note that $\widehat{\mathbf{X}}$ and $\widehat{\mathbf{Y}}$ are intersections of standard $k$-simplices and closed half spaces, furthermore, the intersections are non-empty since by assumption $r_{\text{ave}} \geq \min r_i$ and $j_{\text{ave}} \geq \min j_k$). Therefore, the subspace resulted by the Cartesian product of $\widehat{\mathbf{X}}$ and $\widehat{\mathbf{Y}}$, $( \widehat{\mathbf{X}} \times \widehat{\mathbf{Y}} )$, is also a non-empty, closed and convex subset of $\mathbb{R}^{m+n}$.
Now define $F$, such that,
\begin{equation}
\label{Eq_Set_Valued_Function_The_definition}
 F(\mathbf{x,y}) = F_x(\mathbf{y}) \times F_y(\mathbf{x}) = \{ ( \mathbf{\bar{x}, \bar{y}} ) \}  \quad  \mathbf{(x,y)} \in \widehat{\mathbf{X}} \times \widehat{\mathbf{Y}}
 % **** =\{ \mathbf{\bar x}_i \} \times \{ \mathbf{\bar y}_j\}
\end{equation}
where 
\begin{equation}
\label{Eq_x_i_argmax}
  F_x(\mathbf{y}) = \{ \mathbf{\bar x}_i \} \triangleq \underset{ \mathbf{x} \in \widehat{\mathbf{X}} }{\mathrm{argmax} }\ \mathbf{x}^T A \ \mathbf{y} 
\end{equation}
and 
\begin{equation}
\label{Eq_y_j_argmax}
  F_y(\mathbf{x}) = \{ \mathbf{\bar y}_j \} \triangleq \underset{ \mathbf{y} \in \widehat{\mathbf{Y}} }{\mathrm{argmax} }\ \mathbf{x}^T B \ \mathbf{y} 
\end{equation}
That is, $F$ maps every strategy pair $\mathbf{(x,y)} \in \widehat{\mathbf{X}} \times \widehat{\mathbf{Y}}$ to the Cartesian product of the sets $\{\mathbf{\bar x}_i\}$ and $\{\mathbf{\bar y}_j\}$ (given by $F_x(\mathbf{y})$ and $F_y(\mathbf{x})$, respectively) where all $\mathbf{\bar x}_i$'s are optimal against $\mathbf{y}$ and all $\mathbf{\bar y}_j$'s are optimal against $\mathbf{x}$.

From \eqref{Eq_x_i_argmax} and \eqref{Eq_y_j_argmax} it is clear that for all $\mathbf{(x,y)} \in \widehat{\mathbf{X}} \times \widehat{\mathbf{Y}}$ the set valued function $F_x(\mathbf{y})$ depends only on $\mathbf{y}$ and $F_y(\mathbf{x})$ depends only on $\mathbf{x}$.
Therefore, if we show that $\{\mathbf{\bar x}_i \} = F_x(\mathbf{y})$ is convexed and upper semi-continuous for every $\mathbf{y} \in \widehat{\mathbf{Y}}$ by extending the exact same argument to $\{\mathbf{\bar y}_j \} = F_y(\mathbf{x})$ we can show that $F(\mathbf{x,y})$ is convexed-valued and upper semi-continuous.

Consider $F_x(\mathbf{y})$ in~\eqref{Eq_x_i_argmax}, for any given $\mathbf{y} \in \widehat{\mathbf{Y}}$ the problem in \eqref{Eq_x_i_argmax} is a linear program in $\mathbf{x}$. Therefore, the solution is always at the intersection of some binding constraints, i.e., it is a polytope at some corner of the feasible region in the direction of the gradient of $F_x$ (see Figure~\ref{Fig_Feasible_Region_Partitioning}). As a result, the set $\{ \mathbf{\bar x}_i \}$ is either a singleton in $\widehat{\mathbf{X}}$ 
(when $\nabla F_x(\mathbf{y})$ is not normal to some face of $\widehat{\mathbf{X}}$ -- Figure~\ref{Fig_Feasible_Region_Partitioning}, top) or a face of $\widehat{\mathbf{X}}$ 
(when $\nabla F_x(\mathbf{y})$ is  normal to some face of $\widehat{\mathbf{X}}$ -- Figure~\ref{Fig_Feasible_Region_Partitioning}, bottom), in either case, the solution set is convex and compact for all $\mathbf{y} \in \widehat{\mathbf{Y}}$. By using the same argument, it is clear that $\{\mathbf{\bar y}_j\}$ is also convex and compact for all $\mathbf{x} \in \widehat{\mathbf{X}}$. Therefore, the set-valued function $F( \mathbf{x,y} )$  is also convexed-valued.

It can be shown (by contradiction) that for every sequence $\mathbf{y}_k$ in $\widehat{\mathbf{Y}}$ that converges to $\mathbf{y}$ and for every sequence $\mathbf{\bar x}_k$ that converges to $\mathbf{\bar x}$ such that $\mathbf{\bar x}_k \in F_x(\mathbf{y}_k)$ for all $k\in \mathbb{N}$ then we must have $\mathbf{\bar x} \in F_x(\mathbf{y})$. 
If $F_x$ was not upper semi-continuous then $\mathbf{\bar x} \notin F_x(\mathbf{y})$ for some sequence. 
Assume $F_x{(\mathbf{y})}$ is a singleton in $\widehat{\mathbf{X}}$ (Figure~\ref{Fig_Feasible_Region_Partitioning}, top), we can find $K$ sufficiently large to make $\mathbf{y}_K$ arbitrarily close to $\mathbf{y}$ and therefore, for all $k>K$ we have $F_x(\mathbf{y}_{k>K}) = F_x(\mathbf{y}) = \mathbf{\bar x}$ which is a contradiction. This argument can be easily extended to the case where $F_x(\mathbf{y})$ is some face of $\widehat{\mathbf{X}}$ (Figure~\ref{Fig_Feasible_Region_Partitioning}, bottom).

Therefore, $F_x(\mathbf{y})$ is upper semi-continuous in $\widehat{\mathbf{X}}$ (so is $F_y(\mathbf{x})$ in $\widehat{\mathbf{Y}}$). 
Hence, $F$ is an upper semi-continuous function 
Therefore, the set-valued function $F(\mathbf{x,y})$ defined in \eqref{Eq_Set_Valued_Function_The_definition} satisfies the requirements of Kakutani's theorem and has a fixed point $(\mathbf{x}^*, \mathbf{y}^*)$ such that
\begin{equation}
 ( {\mathbf{x}}^* , {\mathbf{y}}^* )\ \in \  F( {\mathbf{x}}^* , {\mathbf{y}}^* )
\end{equation}
that is, there exist a strategy pair $( {\mathbf{x}}^* , {\mathbf{y}}^* )$ where its elements are optimal against each other and by definition, this is an equilibrium point of $\mathcal{G}$. This concludes the proof.
\end{proof}
% 
% 
% *******************************************************
% 
\begin{figure}
  \centering
  \includegraphics[width = .7\linewidth]{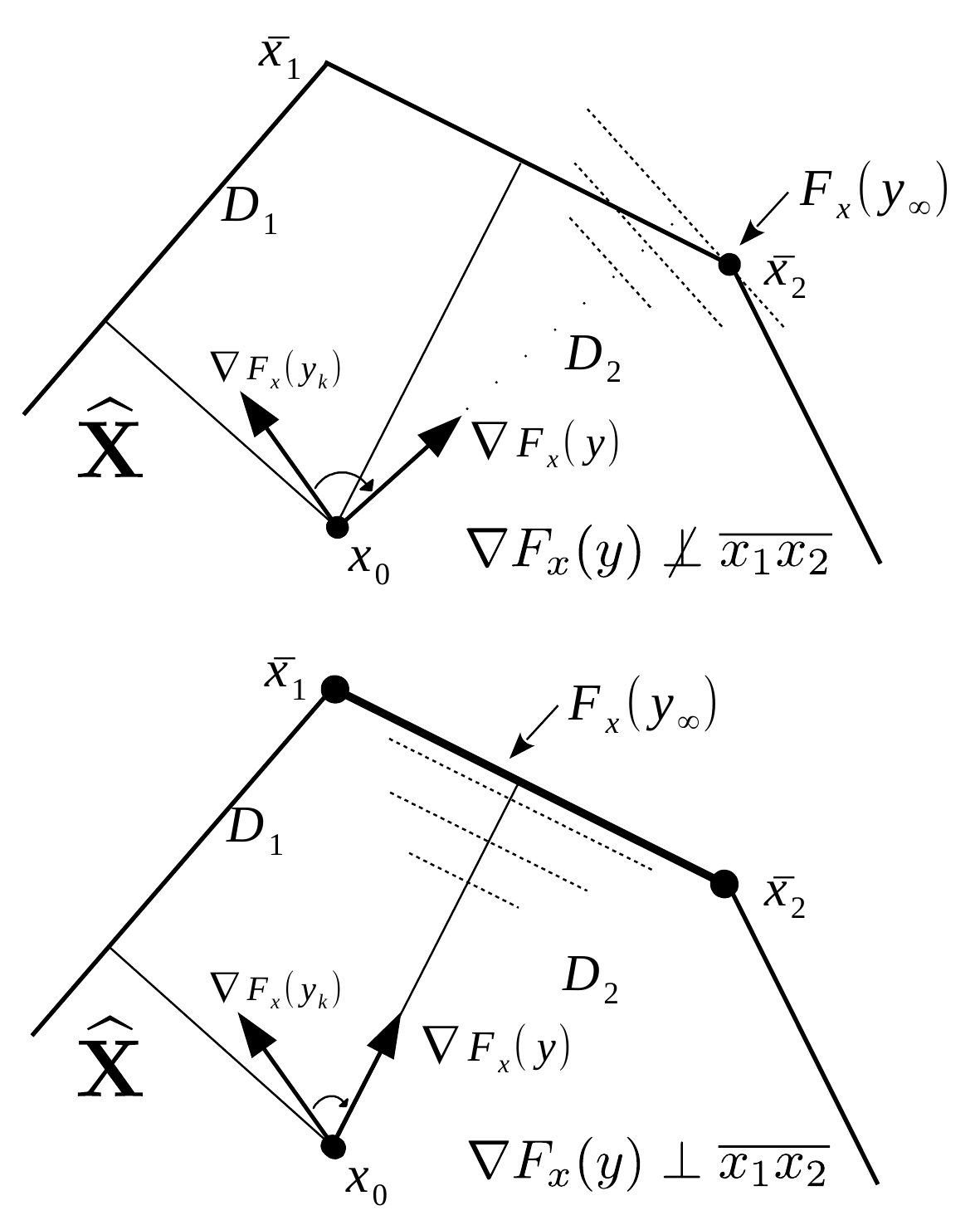} 
  \caption{Geometric representation of upper semi-continuity.}
  \label{Fig_Feasible_Region_Partitioning}
\end{figure}
% 
% *******************************************************
% 
%

% references section
% Generated by IEEEtranTCOM.bst, version: 1.13 (2008/09/30)
% 
\IEEEtriggeratref{24}

\vfill

\end{document}